\documentclass[authoryear,preprint,12pt]{elsarticle}

\makeatletter
\def\ps@pprintTitle{%
	\let\@oddhead\@empty
	\let\@evenhead\@empty
	\let\@oddfoot\@empty
	\let\@evenfoot\@empty
}
\makeatother

\usepackage[utf8]{inputenc}
\usepackage[english,provide=*]{babel}
\usepackage{booktabs}
\usepackage{multirow}
\usepackage{array}
\usepackage{pdflscape}
\usepackage{graphicx}
\usepackage{threeparttable}
\usepackage{caption}
\usepackage{xcolor}
  \usepackage{ulem}
\usepackage{amsthm}
\usepackage{comment}
\theoremstyle{plain}

\newcommand{\Yobs}{Y^{\mathrm{obs}}}
\newcommand{\Eb}{E} 

\usepackage[T1]{fontenc}
\usepackage{amsmath,amssymb,amsthm}
\usepackage[margin=2.5cm]{geometry}
\usepackage{natbib}
\usepackage{tikz}
\usetikzlibrary{shapes.geometric, arrows.meta, positioning}

\newtheorem{assumption}{Assumption}
\newcommand{\indep}{\perp\!\!\!\perp}
\usepackage{algorithm}
\usepackage{algpseudocode}

\usepackage{tikz}
\usepackage{subcaption}
 
\definecolor{todocolor}{HTML}{D81B60}
\definecolor{checks_emmanuel}{HTML}{004D40}
\begin{document}

\begin{frontmatter}

\title{Double Machine Learning for Inequality Decompositions: Reference Outcomes, Orthogonal Scores and Limited Overlap}

\author[label1]{Emmanuel Flachaire}
\author[label2]{Bertille Picard}
\address[label1]{Aix-Marseille Universit\'e, AMSE}
\address[label2]{Univ Rennes, Ensai, CNRS, CREST-UMR 9194, F-35000 Rennes, France}

\begin{abstract}
Kitagawa-Oaxaca-Blinder decompositions split a gap in mean outcomes between two groups into a part explained by characteristics and an unexplained part. The unexplained part depends on the counterfactual ``reference outcome'' that individuals would receive if characteristics were valued equally. Several competing references have been proposed, mostly within the linear regression model. We show that the main proposals, namely using the outcome of either group as the reference and the references of \citet{cotton1988decomposition}, \citet{reimers1983labor} and \citet{neumark1988employers}, all belong to a single family of mixtures of the two groups' potential outcomes. This separates the choice of target from the choice of estimator, and allows nonparametric estimation by double machine learning. We find that the appropriate orthogonal score depends on the reference: Neumark's propensity-weighted reference requires a pooled regression, while all the others require group-specific regressions. We then show that Neumark's reference avoids both extrapolation and extreme weights under limited overlap. Finally, we discuss in a general setting the shrinkage effect of the unexplained part associated with the choice of Neumark's reference.

\end{abstract}


\end{frontmatter}

\newpage

\section{Introduction}

Historically used to decompose gender wage gaps or racial inequalities, the Kitagawa-Oaxaca-Blinder decompositions\footnote{While the decomposition method is also referred to as the “Oaxaca-Blinder decomposition” in the economics literature, similar principles had been introduced earlier by Kitagawa in demography and sociology. The links between these contributions are discussed in \citet{oaxaca2025oaxaca}. Given the generality of our counterfactual framework in this paper, we do not attribute the method to any single author and instead cite all three contributions.} (\citealt{kitagawa1955components,oaxaca1973male,blinder1973wage}) have been widely adopted by empirical researchers. More broadly, these techniques allow for the study of differences in observed outcomes between two groups, including regional comparisons or decompositions of treatment effects in observational contexts. They have also been adopted in Machine Learning (ML) within the field of counterfactual fairness (\citealt{bied2023fairness}).

In the standard context of wage inequality, these techniques help highlight discrimination resulting from a breach of the “equal pay for equal work” principle. This fundamental principle of labor rights in many countries is in itself a justification for using these techniques, which involve asking what individuals would receive if their characteristics as workers were valued equally. Decomposition techniques thus involve counterfactual reasoning, asking what individuals should receive when pay is determined entirely by characteristics. This counterfactual outcome is called the reference outcome. In practice, the choice of the reference outcome has a direct impact on the decomposition: the assessment of the explained and unexplained parts of inequality depends on this choice, particularly in the presence of heterogeneity.

Several proposals have been made to define the reference outcome. In this paper, we group the main proposals from the literature into a general family, defined as the set of reference outcomes equal to a mixture of the potential outcomes of the two groups. We propose several methods for identifying counterfactual quantities: (1) regression, either by predicting the observed outcomes within the subgroups or via a pooled regression of the observed outcome on the characteristics while omitting the group dummy, as suggested by \cite{neumark1988employers}; (2) reweighting using the propensity score, $e(X)=P(D=1\mid X)$; and (3) a combination of regression and reweighting to construct AIPW estimators. AIPW estimators allow us to construct Neyman-orthogonal scores and thus to estimate the nuisance functions (outcomes and propensity scores) using nonparametric methods, following the double machine learning procedure (\citealt{chernozhukov2018double}) to ensure the validity of the confidence intervals.

Since we propose two regression techniques, we obtain two AIPWs and two score functions. We present an original result: the two scores are Neyman-orthogonal under mutually exclusive conditions. Therefore, the choice of the score function must be paired with the choice of the reference outcome. In particular, when the chosen reference is a mixture of the potential outcomes of the two groups weighted by the propensity scores, following the intuition of \cite{neumark1988employers}, it is preferable to use pooled regression.

Finally, we clarify the issues related to limited overlap. These problems are all the more pronounced now that more flexible methods and larger databases can be used to estimate nuisance functions. However, limited overlap has different consequences depending on the chosen estimation method, on the one hand, and the reference outcome, on the other. Most reweighting methods may lead us to divide by the propensity score, which becomes unstable when individuals’ characteristics clearly separate the two groups and the propensity scores become extreme. This problem, however, never arises when the reference chosen is the one inspired by \cite{neumark1988employers}. Regression methods, on the other hand, may lead us to extrapolate into regions where counterfactuals are lacking. Here again, the reference inspired by \cite{neumark1988employers} allows us to avoid this extrapolation: by design, this reference outcome approximates the potential outcome of the dominant group in regions where few alter egos are found in the other group.

Nevertheless, the limited overlap tends to reduce the unexplained share when the chosen reference is that of \cite{neumark1988employers}. This last point had already been mentioned in the literature, in linear, empirical, or theoretical contexts (see \citealt{fortin2008gender}, \citealt{jann2008stata}, \citealt{elder2010}). We therefore investigate this criticism in a fully generalized framework. We demonstrate, in the nonlinear and homogeneous setting, that the unexplained portion based on the reference outcome inspired by \cite{neumark1988employers} is always smaller than the unexplained parts corresponding to other standard reference outcome choices. We show, however, that this result does not hold in the presence of heterogeneity. Nonetheless, empiricists are encouraged to keep this phenomenon in mind when selecting the variables considered to be determinants of the reference outcome: when choosing the reference from \cite{neumark1988employers}, it is imperative not to control for variables correlated with the group membership that are irrelevant to determining the reference outcome. We also discuss what should truly motivate the choice of the reference outcome.

The remainder of the article is structured as follows. First, we describe the problem of decomposing observed inequality within the framework of potential outcomes. We then propose selecting the reference outcome from a generalized family of reference outcomes, which includes the main proposals in the literature. Within this generalized framework, we propose several identification strategies based on the estimation of outcome models, on reweighting via the propensity score, or on a combination of the two, which yields doubly robust estimators. We then discuss parametric and nonparametric approaches based on the double machine learning procedure to evaluate these estimators. In particular, we discuss the conditions for Neyman-orthogonality depending on the choice of reference outcome and of doubly robust estimation strategy. Finally, we discuss the respective advantages of the different reference outcomes, based on the assumptions they allow us to relax. On this occasion, we revisit the criticism by \cite{elder2010} of the reference work by \cite{neumark1988employers}. We show the limits of generalizing this criticism in the presence of heterogeneity, and argue that the choice of the reference outcome must, above all, be motivated by structural economic assumptions.

\section{The decomposition problem}

An individual may be either in group 1 or in group 0. We suppose groups to be two countries, two time periods or two subpopulations (e.g, females and males). Let $D$ be a random variable with support $\mathcal{X}_D=\{0,1\}$ denoting the group of an individual. We assume that the mean outcome is higher in group 1 than in group 0, so that group 1 is considered the advantaged group and group 0 the disadvantaged group. This is not an assumption, but it will be convenient for interpreting the models. Indeed, the labels of the two groups can always be interchanged to check it. One can express the problem in terms of potential outcomes, following the literature on treatment effects \citep{rubin1974estimating}.

An individual in group 1 gets an outcome $Y(1)$ that is observed. In contrast, the outcome $Y(0)$ he would have received if he were in group 0 is not observed. For example, if the individual is a female, we only observe her wage as a female. $Y^\text{obs}$ is a random variable representing the outcome observed depending on the value of $D$ and is supposed to be continuous:\footnote{It is possible to extend the decomposition when $Y^\text{obs}$ is a limited dependent variable. In the case of a binary variable, one can assume an underlying probit model for example, and decompose the probability, as discussed in \citet{FoLeFi:11}, section 3.5.}
\begin{align}
    Y^\text{obs} = Y(1)D + Y(0)(1-D)
    \label{eq:Yobs}
\end{align}

Moreover, we assume that, in each group, the outcome may be explained by some individual characteristics $X$ with support $\mathcal{X}$. They may be unevenly distributed between both categories. For example, men and women may not have the same average level of education. We want to decompose the observed difference in means between a part explained and a part not explained by the characteristics.

In general, the objective is to compare the outcome that individuals in the observed sample receive with an alternative outcome $Y_r$ they could have received. In the simplest case, it can be the outcome received in the other group. Let $\Delta^\text{obs}$ be the expected observed difference in outcomes:
\begin{align}
\Delta^\text{obs}:&= \mathbb{E}[Y^\text{obs} \mid D=1] - \mathbb{E}[Y^\text{obs} \mid D=0] \\
&=\mathbb{E}[Y(1) \mid D=1] - \mathbb{E}[Y(0) \mid D=0] \nonumber
\end{align}

\smallskip
\noindent By adding and subtracting $\mathbb{E}\big[Y_r  \mid  D= 0 \big]$ and $\mathbb{E}\big[Y_r \mid D= 1 \big]$, we obtain
\begin{align}
\Delta^\text{obs}
    &= \underbrace{\mathbb{E}[Y_r \mid D=1] - \mathbb{E}[Y_r \mid D=0]}_{\Delta^r_X} \nonumber \\
    &+ \underbrace{\mathbb{E}[Y(1) \mid D=1] - \mathbb{E}[Y_r \mid D=1]}_{\Delta^{r,1}_S} \nonumber \\
    &- \big( \underbrace{\mathbb{E}[Y(0) \mid D=0] - \mathbb{E}[Y_r \mid D=0]}_{\Delta^{r,0}_S} \big) \label{decompo_refr}  
\end{align}
Thus, a difference in means between two groups can be decomposed in three parts which we denote by:
\begin{itemize}
    \item $\Delta^r_X$ that can be explained by a different composition in terms of characteristics between both groups. For example, one group may have a higher average level of education, leading to a higher average wage.
    \item $\Delta^{r,1}_S$, the unexplained advantage of group 1. 
    \item $-\Delta^{r,0}_S$, the unexplained disadvantage of group 0. 
\end{itemize}
The decomposition of the observed difference in outcomes is sensitive to the choice of the reference outcome.

A standard choice for $Y_r$ is $Y(0)$. Taking $Y_r:=Y(0)$ implies that, in a world where individuals from both groups with the same characteristics receive identical outcomes, everyone obtains the potential outcome corresponding to the disadvantaged group $Y(0)$. In this case, $\Delta_S^{0,0}=0$. This leads to the decomposition:
\begin{align}
    \Delta^\text{obs}&= \underbrace{\mathbb{E}[Y(0) \mid D=1] - \mathbb{E}[Y(0) \mid D=0]}_{\Delta^0_X} 
    +\underbrace{\mathbb{E}[Y(1) \mid D=1] - \mathbb{E}[Y(0) \mid D=1]}_{\Delta^{0,1}_S\text{ (ATT)}} \label{decompo_ref0}
\end{align}
Note that this is exactly the decomposition derived in \citet{angrist2009mostly}, specifically in equation (3.2.1) of their book, in terms of Average Treatment effects on the Treated, ATT$:=\mathbb{E}[Y(1)-Y(0) \mid D=1]$, and selection bias. In the context of inequality decomposition, $\Delta_S^{0,1}$ is not a treatment effect, since $D$ is not a treatment. The analogy is nonetheless useful: it explains why readers familiar with the treatment effects literature will recognize the estimators presented below. It also shows how to use this framework when working with treatment effects: when $D$ is a treatment, our framework decomposes the gap between treated and untreated individuals, and the ATT and the ATU are the unexplained parts associated with two particular choices of reference outcome ($w=0$ and $w=1$, respectively). This helps clarify why the ATT and the ATU generally differ: they are ``unexplained'' parts, in the sense that they are not explained by baseline characteristics, corresponding to different reference outcomes.

Alternatively, one could consider that when people face no unexplained observed difference in means on the market, they all receive the outcome of the advantaged group, thus $Y_r := Y(1)$ and $\Delta^{1,1}_S=0$. This leads to the following decomposition:
\begin{align}
    \Delta^\text{obs}&= \underbrace{\mathbb{E}[Y(1) \mid D=1] - \mathbb{E}[Y(1) \mid D=0]}_{\Delta^{1}_X} 
    - \big( \underbrace{\mathbb{E}[Y(0) \mid D=0] - \mathbb{E}[Y(1) \mid D=0]}_{\Delta^{1,0}_S \text{ (ATU)}} )  \label{decompo_ref1}
\end{align}
Here, the unexplained part of the difference in means is analogous to an Average Treatment effect on the Untreated, ATU$:=\mathbb{E}[Y(1)-Y(0) \mid D=0]$.


\section{A generalized reference outcome}

The literature proposes several decomposition methods. These proposals have mostly been formulated within the linear regression model. For example, the R package \texttt{oaxaca} provides estimates based on standard methods (selecting one of the two groups as the reference) and on the methods proposed by \citet{cotton1988decomposition}, \citet{reimers1983labor}, and \citet{neumark1988employers}. Each of them can be associated with a target value, defined by a choice of reference outcome. We propose a general framework that describes all these approaches as weightings of the potential outcomes of the two groups, with weights summing to 1. This will allow us, in the following sections, to define very general formulas as well as a comprehensive set of estimation techniques, and to move beyond the standard linear framework that has dominated this literature.

In order to provide a general framework for identification that is compatible with most of the proposed reference outcomes in the literature, let us define a family of reference outcomes $\mathcal{F}$ that includes them. This family of reference outcomes includes any reference outcome corresponding to a weighting of $Y(1)$ and $Y(0)$, whose weights are denoted by $w(X)$ and $(1-w(X))$, where $w(X)$ may be a constant or depend on the characteristics $X$\footnote{Note that $w(X)$ is a (known) measurable function of $X$ alone, so that $\Eb[Z\,w(X)\mid X]=w(X)\,\Eb[Z\mid X]$ for any integrable random variable $Z$.}:
\[
\mathcal{F} = \left\{\, Y_w = w(X)\,Y(1) + (1-w(X))\,Y(0) \;:\; w \in \mathcal{W} \,\right\},
\qquad \mathcal{W} = \left\{\, w : \mathcal{X} \to [0,1] \text{ measurable} \,\right\},
\]
where each reference outcome is indexed directly by its weight function $w$. Throughout, $Y(d)$, $d\in\{0,1\}$, denotes a potential outcome, while $Y_w$ denotes a reference outcome.

This family includes the most common reference outcomes:
\begin{itemize}
\item $Y_0=Y(0)$. This case is associated to the choice $w(X)=0$ and implies that the reference outcome is the outcome of the disadvantaged group.
\item $Y_1=Y(1)$. This case is associated to the choice $w(X)=0$ and implies that the reference outcome is the outcome of the advantaged group.
\item $Y_e= Y(1)\,e(X) + Y(0)\,(1 - e(X))$. This case includes \citet{neumark1988employers}'s proposal. In his paper, Neumark restricts himself to the case where $X$ is a categorical variable with a small number of categories. Here, we extend this definition to continuous variables and categorical variables with a very large number of categories, in order to be able to work in a high-dimensional setting. $Y_e$ can be interpreted as an equilibrium outcome that pools the two group-specific models. In the context of the labor market, it corresponds to the wage structure that would prevail under an equilibrium model of employer behavior (see Section~\ref{sec:need_for_an_econ_theory} for a discussion).
\item $Y_\pi= \pi Y(1) + (1-\pi) Y(0)$, where $\pi = \mathbb{P}(D=1)$. This corresponds to \citet{cotton1988decomposition}'s proposal.
\item $Y_{1/2}= \frac{1}{2} Y(1) + \frac{1}{2} Y(0)$. This corresponds to \citet{reimers1983labor}'s proposal.
\end{itemize}

Let $\tau(X)=Y(1)-Y(0)$. For the reference outcomes of the family $\mathcal{F}$, the unexplained part is $\delta_w$:
\begin{align}
\label{unexplained_part_as_function_of_tau}
\delta_w = E\big[(1-w(X))\,\tau(X) \mid D=1\big] + E\big[w(X)\,\tau(X) \mid D=0\big]
\end{align}
where $\tau(X) = Y(1) - Y(0)$. This equality is proven in Appendix~\ref{appendix:inexp-Fw-1}. The unexplained part is thus a weighted combination of the gaps $Y(1)-Y(0)$ in the two groups: individuals of group~1 contribute with weight $1-w(X)$ and individuals of group~0 with weight $w(X)$. The case $w=0$ gives the ATT, the case $w=1$ gives the ATU, and any $w$ in between mixes the two. Equivalently,
\begin{align}
\label{unexplained_part_unconditional}
\delta_w = E\!\left[\left(\frac{(1-w(X))\,D}{\pi} + \frac{w(X)\,(1-D)}{1-\pi}\right)\tau(X)\right],
\end{align}
which is the form used in the proofs.

The problem with identifying these quantities lies in the terms:
\begin{align*}
\tau(X)\,D &= Y(1)D - Y(0)D = Y^\text{obs}D - Y(0)D,\\
\tau(X)(1-D) &= Y(1)(1-D) - Y(0)(1-D) = Y(1)(1-D) - Y^\text{obs}(1-D)
\end{align*}
Each contains an unobservable term generally referred to as a counterfactual: $Y(0)$ in the group $D=1$ and $Y(1)$ in the group $D=0$.

\section{Identification strategies}
Let's rewrite $\delta_w$ to clearly visualize the counterfactuals we need to identify:
\begin{align}
\delta_w &= E\!\left[\frac{(1-\pi)(1-w(X))\,Y^\text{obs}D - \pi\, w(X)\,Y^\text{obs}(1-D)}{\pi(1-\pi)}\right] \nonumber\\
&\quad - E\!\left[\frac{(1-\pi)(1-w(X))\,Y(0)D - \pi\, w(X)\,Y(1)(1-D)}{\pi(1-\pi)}\right]
\label{dec:unexpl:as_sum_obs_and_count}
\end{align}
where the counterfactual quantities all appear in the second expectation.

The proposed identification strategies are based on the assumption that:
\begin{assumption} \label{as:ignorability}
{\textbf{[ignorability]}}
\begin{align*}
(Y(0),Y(1)) \indep D \mid X 
\end{align*}
where $\indep$ denotes statistical independence \citep{rosenbaum_central_1983}. It is also known as the {\em conditional independence} assumption or {\em unconfoundedness} \citep{imbens2004nonparametric}.
\end{assumption}

\subsection{Regression}

We propose two identification strategies based on regressions of the outcome models to estimate the unexplained part associated to any reference outcome from family $\mathcal{F}$. The first is based on two groupwise regressions; the second uses a pooled regression and an assumption about the propensity score model.

\subsubsection{Main assumptions}

The regression functions $\mu_d(X)=E[Y^\text{obs}\mid X, D=d]$ within the subgroups identify the counterfactuals and, therefore, the unexplained part. Alternatively, the unexplained part can be identified from the pooled regression function $m(X)=E[Y^\text{obs}\mid X]$. To do this, we will thus make the following assumptions:

\begin{assumption} \label{as:g1}
{\textbf{[groupwise 1]}}
$E[Y(1)\mid X, D=1] = \mu_1(X)$
\end{assumption}
\begin{assumption} \label{as:g0}
{\textbf{[groupwise 0]}}
$E[Y(0)\mid X, D=0] = \mu_0(X)$
\end{assumption}
\begin{assumption} \label{as:g2}
{\textbf{[pooled]}}
$E[Y^\text{obs}\mid X] = m(X)$
\end{assumption}

In all linear regression approaches, the functions $\mu_0(X)$, $\mu_1(X)$ and $m(X)$ are assumed to be linear combinations of the features $X$. More generally, we will treat $\mu_0$, $\mu_1$ and $m$ as functions to be estimated, whose form may be nonlinear, in order to increase the flexibility of our models.
Furthermore, we will assume an overlap assumption:

\begin{assumption} \label{as:overlap}
{\textbf{[overlap]}}
$0 < P(D=1\mid X) < 1$
\end{assumption}

This assumption ensures that there are alter egos in each group for every profile $X$. Finally, the pooled regression approach described below may also rely on a model for the propensity score, which we denote $e(X)$:

\begin{assumption} \label{as:propensityscore_model}
{\textbf{[propensity score]}}
$P(D=1\mid X)=e(X)$
\end{assumption}

\subsubsection{A regression-based estimator depending on $\mu_1(X)$ and $\mu_0(X)$}

Under Assumptions~\ref{as:ignorability}, \ref{as:g1} and \ref{as:g0}, we can rewrite $\delta_w$:
\begin{align}
\delta_w = E\!\left[\frac{(1-w(X))\,D\,(Y^\text{obs} - \mu_0(X))}{\pi} - \frac{w(X)\,(1-D)\,(Y^\text{obs} - \mu_1(X))}{1-\pi}\right]
\label{unexp:regression_g0_g1}
\end{align}

The proof is given in Appendix~\ref{appendix:inexp-Fw-3}. The estimators corresponding to each reference outcome, that is, to each choice $w(X)$, are detailed in Table~\ref{tab:regression_estimators} (first column of the regression approach). 

\subsubsection{A regression-based estimator depending on $m(X)$}

Under Assumptions~\ref{as:ignorability}, \ref{as:g2}, \ref{as:overlap} and \ref{as:propensityscore_model}, we can rewrite $\delta_w$ as:
\begin{align}
\delta_w = E\!\left[\frac{(1-w(X))}{\pi(1-e(X))}\,D\,(Y^\text{obs} - m(X)) - \frac{w(X)}{(1-\pi)\,e(X)}\,(1-D)\,(Y^\text{obs} - m(X))\right]
\label{unexp:regression_g2}
\end{align}

The proof is given in Appendix~\ref{appendix:inexp-Fw-2}.

The estimators corresponding to each reference outcome, that is, to each choice $w(X)$, are detailed in Table~\ref{tab:regression_estimators} (second column of the regression approach). Note that this estimator is unusual for cases $w(X)=0$ and $w(X)=1$ (or, similarly, when the ATT or the ATU is the target in the treatment effect literature). It is easy to understand why they are not very popular reading the second column of Table~\ref{tab:aipw_estimators}: they both require a division by the propensity score, in addition to an estimation for $m(X)$. Not only this adds an estimation step, it also generates an estimator unstable when the denominator approaches $0$. Nonetheless, we will see below that this form is useful to build an appropriate AIPW estimator, particularly for the case $w(X)=e(X)$, implying that $Y_e$ will be the reference outcome.

\subsection{Inverse Probability Weighting}

An alternative to the regression approach is the reweighting approach using the propensity score $e(X)$. Let us return to formula~\eqref{dec:unexpl:as_sum_obs_and_count} to propose a strategy for identifying the second expectation. Under Assumption~\ref{as:propensityscore_model} and the overlap assumption (assumption~\ref{as:overlap}), the IPW estimator of the unexplained component is based on the equality:
\begin{align}
\label{IPW}
    \delta_w = E\!\left[\left(\frac{1-w(X)}{\pi\,(1-e(X))} + \frac{w(X)}{(1-\pi)\,e(X)}\right) Y^\text{obs}\,(D - e(X))\right]
\end{align}

All IPW estimators for each choice of function $w(X)$, and thus of reference, are presented in Table~\ref{tab:ipws}.

\subsection{Augmented Inverse Probability Weighting}

We can also derive doubly robust estimators simultaneously based on regression and the inverse probability weighting estimator. A doubly-robust estimator is consistent if at least one of the models is correctly specified, either the outcome model or the probability model \citep{robins1995semiparametric}. Thus, it combines both strategies. The Augmented Inverse Probability Weighting (AIPW) estimator is the most commonly known doubly robust estimator, and typically relies on parametric regressions to estimate the outcome model and the propensity score.\footnote{In fact, \citet{kline2011oaxaca} shows that the standard Kitagawa-Oaxaca-Blinder estimator constitutes a propensity score re-weighting estimator based upon a linear model for the conditional odds of being treated. This includes assignment models with a latent log-logistic error structure and may yield negative weights. As such, it is a ``doubly robust'' estimator. Nevertheless, most presentations of doubly robust estimators assume that the propensity score model and the plug-in in the estimator will be explicitly re-estimated.} 

\subsubsection{AIPW estimators based on $\mu_1(X)$ and $\mu_0(X)$}

The general formula for the AIPW estimator based on the estimations of $\mu_1(X)$ and $\mu_0(X)$ is:

\begin{align}
\label{AIPW_g0_g1}
\delta_w &= E\!\left[\frac{1}{\pi}\,(Y^\text{obs}-\mu_0(X))\,(D-e(X))\,\frac{1-w(X)}{1-e(X)}\right] \nonumber \\
 &+ E\!\left[\frac{1}{1-\pi}\,(Y^\text{obs}-\mu_1(X))\,(D-e(X))\,\frac{w(X)}{e(X)}\right]
\end{align}

This formula relies on ignorability (Assumption~\ref{as:ignorability}), Assumption~\ref{as:g1}, Assumption~\ref{as:g0}, Assumption~\ref{as:overlap} and Assumption~\ref{as:propensityscore_model}. Proof is left in Appendix~\ref{appendix:aipw_mu0_mu1}.

The estimators associated to classical references $Y_w$ associated to the different functions $w(X)$ are detailed in the first column of Table~\ref{tab:aipw_estimators}. Note that from this formula, we can recover the standard AIPW estimators for the ATT and ATU (see case $w(X)=1$ and $w(X)=0$).

\subsubsection{AIPW estimators based on $m(X)$}

The general formula for the AIPW estimator based on the estimation of $m(X)$ is:

\begin{align}
\label{AIPW_g2}
\delta_w = {}& E\!\left[\left(\frac{1-w(X)}{(1-e(X))\,\pi}+\frac{w(X)}{e(X)\,(1-\pi)}\right)(D-e(X))\,(Y^\text{obs}-m(X))\right]
\end{align}

This formula relies on ignorability (Assumption~\ref{as:ignorability}), Assumption~\ref{as:g2}, Assumption~\ref{as:overlap} and Assumption~\ref{as:propensityscore_model}. Proof is left in Appendix~\ref{appendix:aipw-g2}.

The estimators associated to classical references $Y_w$ associated to the different functions $w(X)$ are detailed in the second column of Table~\ref{tab:aipw_estimators}. Note that from this formula, we can derive AIPW estimators for the ATT and ATU based on a pooled regression of $Y^\text{obs}$ on $X$ (see case $w(X)=1$ and $w(X)=0$). Again, this estimator yields estimators that are not commonly used in the treatment effects literature for the ATT and ATU.


\section{Estimation}

Empirical counterparts of these estimators are straightforward and rely on plug-in strategy of the unknown nuisance functions: either the pair $(\mu_0(X), \mu_1(X))$ or $m(X)$ for the outcome models, and/or the propensity score $e(X)$. 

\subsection{Parametric approach}

The specific forms of the estimators (\ref{unexp:regression_g0_g1}), (\ref{IPW}), and (\ref{AIPW_g0_g1}) in the cases $w(X)=0$ and $w(X)=1$ have long been known in the treatment effects literature (see \citealt{imbens2004nonparametric}, or \citealt{imbens_causal_2015} for a textbook). Traditionally, the outcome models are estimated under a linear specification via Ordinary Least Squares (OLS), while the propensity score is estimated using a logit or probit regression. Parametric regression models can be estimated efficiently under standard regularity assumptions, which makes this approach particularly attractive: statistical inference is then feasible at a sufficient rate, allowing us to rely on valid confidence intervals. This is generally not the case when other algorithms are used instead. However, the unbiasedness and consistency of these estimators depend on the model being correctly specified. Moreover, linear regression can only be estimated by OLS when $X'X$ is invertible, which requires the number of covariates not to exceed the number of observations, which is a condition violated in high-dimensional settings.


\subsection{Machine learning approach}
\label{sec:ML-approach}

Supervised machine learning methods predict an outcome $Y^\text{obs}$ using covariates $X$. Here, we will use the term machine learning methods to refer to classical methods such as lasso, ridge, random forests, boosting, or neural networks (see \citealt{hastie2005elements}). Linear regression estimated by OLS and logit/probit models estimated by maximum likelihood may be presented as basic methods in standard textbooks. However, the term ``machine learning'' more generally refers to flexible methods capable of capturing nonlinear relationships between covariates and the outcome, interaction effects between variables on the outcome, and/or selecting relevant covariates for models. These methods can therefore be used to predict $\mathbb{E}\big[Y^\text{obs}\mid X\big]$, $\mathbb{E}\big[Y^\text{obs}\mid X, D=d \big]$ for any $d\in\{0,1\} $, or $\mathbb{P}\big(D=d \mid X\big)$ for any $d\in\{0,1\}$. This is why they can be considered as alternatives or extensions to the parametric methods presented in the previous sections.

Non-parametric methods are often based on a bias/variance tradeoff and typically provide biased estimates with slower convergence rates than parametric models. This may not be an issue when the focus is on prediction, but it may be problematic when the quality of the estimation is key, and inference is required. In our application, the parameter of interest is not directly the quantity predicted by supervised machine learning techniques. Rather, we aim to estimate it using plug-in strategies based on either an estimation for the outcome model, the propensity score model, or both. Since these naive plug-in approaches neither guarantee convergence to the target value at a sufficient rate nor allow for the construction of valid confidence intervals or the asymptotic normality of the estimators, alternative strategies have been proposed in the literature to adapt plug-in methods estimated via machine learning while preserving desirable properties for statistical inference. We focus here on the Double Machine Learning procedure \citep{chernozhukov2018double}.

\subsubsection{Double Machine Learning Procedure}

\citet{chernozhukov2018double} developed a theory of inference with machine learning methods and proposed a double/debiased Machine Learning method providing estimators that can achieve $\sqrt{n}$-consistency and asymptotic normality. 
Let $\theta^*$ be the true value of a low-dimensional parameter of interest and $\eta^*$  the true nuisance parameter. 
In our decomposition framework, $\theta^*$ would be $\delta_w$  and $\eta^*$ would be $(\mu_0, \mu_1, e)$ or $(m, e)$, depending on the estimator.  The double/debiased machine learning is based on two crucial ingredients: the Neyman orthogonality condition and cross-fitting. 

\medskip
{\bf Neyman orthogonality}: Let us consider scores $\psi(.)$ that satisfy the following identification condition:
\begin{equation}
    \mathbb{E} [\psi(W;\theta^*,\eta^*)] = 0
\label{eq:moment}
\end{equation}
and the Neyman orthogonality condition
\begin{equation}
    \frac{\partial}{\partial c} \mathbb{E} [\psi(W;\theta^*,\eta^*+c(\eta-\eta^*)] \Big\vert_{c=0} = 0
    \label{eq:orthogonality}
\end{equation}
where $W=(Y^{\mathrm{obs}}, D, X)$ and $c\in [0,1)$.  An estimator of the parameter of interest $\theta^*$ can be obtained from the empirical analog of (\ref{eq:moment}) with $\eta^*$ replaced by a machine learning estimation $\hat\eta$.
The orthogonality condition (\ref{eq:orthogonality}) ensures the estimator of $\theta^*$ to be insensitive to small mistakes in the estimation of nuisance parameters. 
 \citet{chernozhukov2018double} show that, with  scores defined by (\ref{eq:moment}) and (\ref{eq:orthogonality}) and nuisance parameters quite well estimated by $\hat\eta$,\footnote{It is required that $\hat\eta$ converges at a rate faster than $n^{-1/4}$. These rates are achieved by many ML methods under suitable regularity conditions.} the estimator of $\theta^*$ obtained from
\begin{equation}
    \frac{1}{n} \sum_{i=1}^n  \psi(W_i;\hat\theta,\hat\eta) = 0
    \quad\text{is asymptotically Normal} \quad
    \sqrt{n} ( \hat\theta - \theta^*) \overset{d}{\longrightarrow} \mathcal{N} (0, \sigma^2)
\label{eq:thetahat}
\end{equation}
where  $\sigma^2=\mathbb{E} [\psi^2(W_i;\theta^*,\eta^*)]$ and $n$ is the total number of observations in the sample. 

\medskip
{\bf Cross-fitting}: The parameter of interest and the nuisance parameters should be estimated from two different samples. Let us consider $K$ replications, where an estimator of $\theta^*$ is obtained on a random subsample $I_k$ with $n_k$ observations from
\begin{equation}
    \frac{1}{n_k} \sum_{i\in I_k}  \psi(W_i;\hat\theta_k,\hat\eta_k) = 0
\label{eq:thetaKhat}
\end{equation}
while the nuisance functions $\hat\eta_k$ are estimated from an auxiliary sample, which includes all observations that are not in $I_k$.
This can be done for each replication. Finally, the parameter of interest $\theta^*$ can be estimated by averaging all $\hat\theta_k$
\begin{equation}
\hat\theta = \frac{1}{K} \sum_{k=1}^K \hat\theta_k \label{eq:thetahat-2}
\end{equation}
and the variance of $\hat\theta$ can  be estimated with $\hat\sigma^2/n$ where
\begin{equation}
    \hat\sigma^2 = 
    \frac{1}{K} \sum_{k=1}^K \frac{1}{n_k}   \sum_{i\in I_k}  \psi^2(W_i;\hat\theta_k,\hat\eta_k)
\label{eq:variance}
\end{equation}
Hence, the standard error of $\hat\theta$ is 
\begin{equation}
    \text{se}(\hat\theta) = \frac{\hat{\sigma}}{\sqrt{n}}
\label{eq:se}
\end{equation}
More generally, when the score is linear in $\theta$, $\psi(W;\theta,\eta)=\psi^a(W;\eta)\,\theta+\psi^b(W;\eta)$, the asymptotic variance is $\sigma^2=\mathbb{E}[\psi^2(W;\theta^*,\eta^*)]/J^2$ with $J=\mathbb{E}[\psi^a(W;\eta^*)]$ \citep[Theorem~3.1]{chernozhukov2018double}, and $\hat\sigma^2$ in~(\ref{eq:variance}) is divided by $\hat J^2$, the cross-fitted empirical counterpart of $J^2$; equation~(\ref{eq:variance}) corresponds to $\psi^a=-1$. Equivalently, $\sigma^2=\mathbb{E}[\varphi^2]$, where $\varphi=-\psi/J$ is the influence function of $\hat\theta$. 

Hence, despite relying on machine learning estimators of $\eta^*$ that converge at a slower-than-$\sqrt{n}$ rate, the resulting estimator $\hat\theta$ retains $\sqrt{n}$-consistency and asymptotic normality, allowing for the construction of valid confidence intervals for $\delta_w$.

\subsubsection{Score function based on $\mu_0(X)$ and $\mu_1(X)$ for $w(X)$ constant}

We can naturally construct a score function that satisfies the equality~(\ref{eq:moment}) by starting from the AIPW estimator based on $\mu_0(X)$ and $\mu_1(X)$ defined in equation~(\ref{AIPW_g0_g1}): 

	\begin{align}
    \label{score:aipw_g0g1}
		\psi_1(W;\theta,\eta)
		&= \frac{1}{\pi}\,\big(\Yobs-\mu_0(X)\big)\,
		\big(D-e(X)\big)\,\frac{1-w(X)}{1-e(X)} \\[4pt]
		&\quad+\ \frac{1}{1-\pi}\,\big(\Yobs-\mu_1(X)\big)\,
		\big(D-e(X)\big)\,\frac{w(X)}{e(X)} -\ \theta \nonumber
	\end{align}

It is known from \citet{chernozhukov2018double} that the score function based on the AIPW estimator in~(\ref{AIPW_g0_g1}) based on $\mu_0(X)$ for $w(X)=0$ satisfies the Neyman orthogonality condition (the estimator is then analogous to an ATT). Symmetrically, this is also true for $w(X)=1$. We extend this result by proving that the general estimator in~(\ref{AIPW_g0_g1}) satisfies the Neyman orthogonality condition for any $w(X)$ constant in Appendix~\ref{proof:score_g0g1}. This allows us to exploit this score function to extend the approach in \cite{cotton1988decomposition} and \cite{reimers1983labor}. Indeed, the AIPW estimators can then be used with machine learning methods and cross-fitting to estimate the unexplained part of the observed difference in outcomes, as detailed in Algorithm \ref{algo:DR-constant}. 

Importantly, we show that this condition is not verified in general for $w(X)=e(X)$, hence for $Y_e$ as the reference, in Appendix~\ref{proof:score_g0g1}. We provide concrete guidelines for implementing this estimation strategy in Algorithm~\ref{algo:DR-constant}.

\subsubsection{Score function based on $m(X)$ for $w(X)=e(X)$}

We can construct another score function based on the AIPW estimator defined by $m(X)$ in equation~(\ref{AIPW_g2}):
\begin{align}
\label{score:aipw_g2}
\psi_2(W;\theta,\eta) = \left(\frac{1-w(X)}{(1-e(X))\,\pi}+\frac{w(X)}{e(X)\,(1-\pi)}\right)\big(D-e(X)\big)\,\big(\Yobs-m(X)\big) - \theta,
\end{align}

In Appendix~\ref{appendix:score-aipw-g2}, we show that this score function satisfies the Neyman orthogonality condition for $w(X)=e(X)$, but does not satisfy it for a constant $w(X)$. In other words, the conditions for Neyman orthogonality of the score functions $\psi_1$ and $\psi_2$ are mutually exclusive. Since Neyman orthogonality is what guarantees valid inference despite machine-learning-based nuisance estimation, we recommend using the AIPW estimator based on $\mu_0(X)$ and $\mu_1(X)$ for a constant $w(X)$ and $m(X)$ for $w(X)=e(X)$, unless future work establishes valid inference under conditions weaker than Neyman orthogonality.  We provide concrete guidelines for implementing this estimation strategy in Algorithm~\ref{algo:DR-neumark}.

\subsubsection{Main take-away: mutually exclusive orthogonality}

In summary, the conditions for Neyman orthogonality of the two scores are mutually exclusive, and the score must be paired with the correct choice of reference outcome:

\begin{center}
\begin{tabular}{@{}lcc@{}}
\toprule
 & $\psi_1$ & $\psi_2$  \\
\midrule
constant $w(X)$ (e.g.\ $w=0,1,\pi,\tfrac12$) & orthogonal & not orthogonal \\
$w(X)=e(X)$ & not orthogonal & orthogonal \\
\bottomrule
\end{tabular}
\end{center}

We present a simulation exercise to illustrate this result. The absence of Neyman orthogonality does not necessarily imply bias that persists to infinity: whether $w(X)$ is constant or equal to $e(X)$, the two proposed AIPW estimators remain consistent as long as the nuisance functions are consistent. The difference lies in their convergence rates. The simulation results shown in Figure~\ref{fig:NOrt} illustrate this. Subfigure~\ref{fig:NOrt_a} shows the estimates, obtained via gradient boosting, of the nuisance functions $\mu_0(X)$, $m(X)$, and $e(X)$, where $X$ is a continuous variable. We observe that the estimates exhibit a slight bias, even for a large sample size (here $n=8000$). Subfigure~\ref{fig:NOrt_b} shows the bias multiplied by $\sqrt{n}$, that is, the bias expressed in units of standard deviation, the quantity that directly determines the coverage. We observe that for the reference $Y(0)$, the choice of the AIPW based on $\mu_0(X)$ and $\mu_1(X)$ (labelled \texttt{g0} and \texttt{g1} in the figure) is preferable, since this quantity remains close to $0$ regardless of the sample size. If we instead choose the AIPW based on $m(X)$ (labelled \texttt{g2} in the figure), this quantity actually moves away from $0$ as $n$ increases: although the bias decreases in absolute value with sample size, it decreases too slowly compared to the rate $\sqrt{n}$ at which the confidence interval narrows. The risk, therefore, is that confidence intervals will become increasingly narrow around a value that is slightly offset from the true value, even as the estimator approaches the true value as the sample size increases. This is directly illustrated in subfigure~\ref{fig:NOrt_c}, where the coverage drops when the AIPW based on $m(X)$ is chosen for the reference $Y(0)$. As predicted, exactly the opposite occurs when $Y_e$ is chosen as the reference: in that case, the AIPW based on $m(X)$ is the most appropriate.

\begin{figure}
     \centering
        \caption{Illustration of the mutually exclusive conditions for Neyman orthogonality (over 200 simulations)}
        \label{fig:NOrt}
     \begin{subfigure}[b]{\textwidth}
         \centering
         \caption{Estimations by a gradient boosting algorithm of the nuisance functions on one simulation}
         \includegraphics[width=\textwidth]{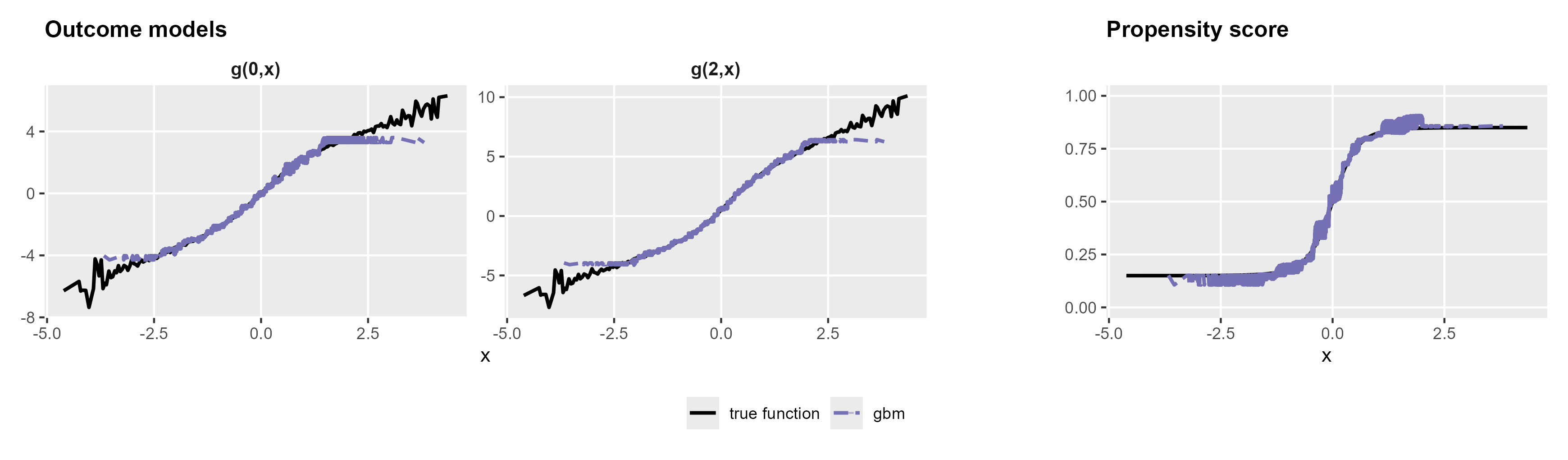}
         \label{fig:NOrt_a}
     \end{subfigure} \\
     \begin{subfigure}[b]{0.9\textwidth}
         \centering
         \caption{Bias times $\sqrt{n}$ when using appropriate score functions (green and continuous lines) and inappropriate score functions (red and discontinuous lines)}
         \includegraphics[width=\textwidth]{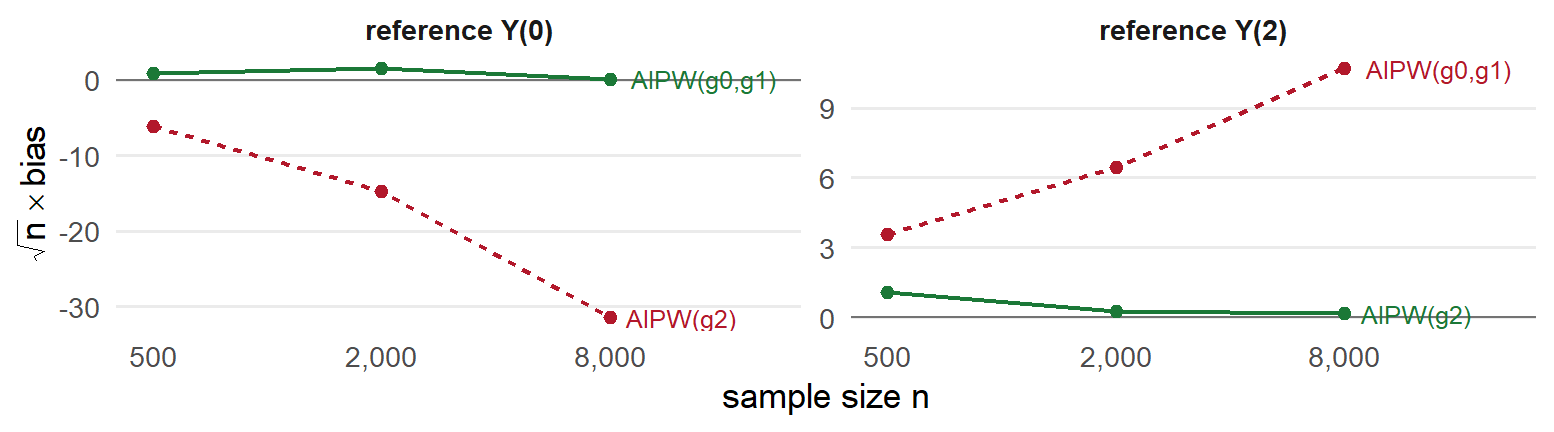}
         \label{fig:NOrt_b}
     \end{subfigure}
     \\
     \begin{subfigure}[b]{0.7\textwidth}
         \centering
         \caption{Decline of the coverage when using appropriate score functions (green and continuous lines) and inappropriate score functions (red and discontinuous lines)}
         \includegraphics[width=\textwidth]{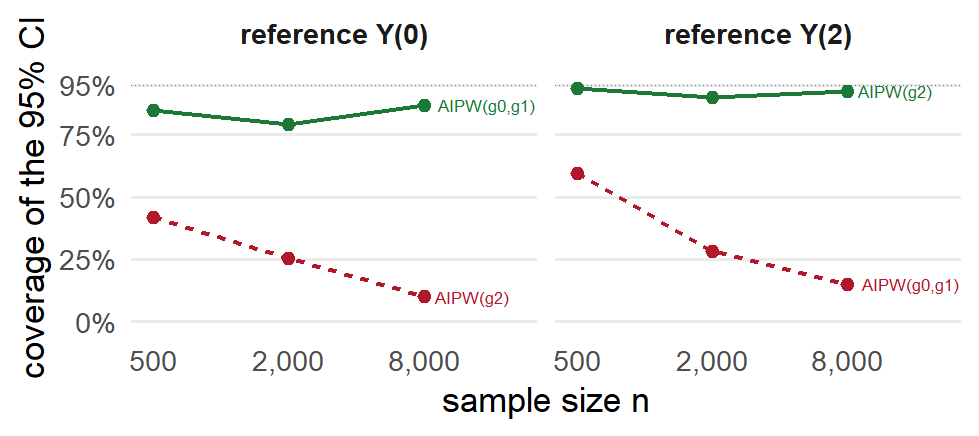}
         \label{fig:NOrt_c}
     \end{subfigure}
\end{figure}

\section{When does the choice of reference outcome matter?}

\subsection{Assumption savings when choosing $w(X)=0$ or $w(X)=1$}
\label{sec:assumption_savings}

If $w(X)=0$, the second term of expression~(\ref{unexplained_part_as_function_of_tau}) disappears, and only a single counterfactual needs to be estimated. We can then relax some assumptions. For the estimators~(\ref{unexp:regression_g0_g1}) and (\ref{AIPW_g0_g1}), Assumption~\ref{as:g1} is no longer necessary. For the estimators~(\ref{unexp:regression_g2}), (\ref{IPW}) and (\ref{AIPW_g2}), Assumption~\ref{as:overlap} can be relaxed and replaced by $P(D=1\mid X) < 1$, i.e., a constraint on one side of the propensity score only. This can be useful in cases where certain individuals in the $D=0$ group have very few counterparts in the $D=1$ group. In other words, it is convenient if certain profiles $X$ are in the $D=0$ group and never appear in the $D=1$ group, such that $P(D=1\mid X)$ is close to $0$ for these $X$. Similarly, we make analogous simplifications in the case where $w(X)=1$. This time, it is the other counterfactual that no longer needs to be estimated.

\subsection{Stability under limited overlap using $w(X)=e(X)$}

A situation of limited overlap is one in which Assumption~\ref{as:overlap} holds strictly in a theoretical sense, but in practice, the propensity score takes on extreme values, close to $0$ and/or $1$. This makes the estimators~(\ref{unexp:regression_g2}), (\ref{IPW}), (\ref{AIPW_g0_g1}), and (\ref{AIPW_g2}) particularly unstable. This type of situation is all the more plausible when working with large datasets. Thus, using a large data set of 1.7 million employees in Switzerland, \citet{strittmatter2021gender} find that with the 10 most important variables for explaining male wages, 55\% of women in the private sector have no comparable men with respect to these variables. A common practice in this situation in the treatment effect literature is to perform trimming. It consists in removing observations without (enough) pairs in the reference group. 

Performing such trimming when the reference outcome is $Y(0)$ involves considering the following reference outcome:
\begin{equation}
\text{Trimming:} \quad Y_r= 
\begin{cases}
Y(0) & \text{if \,\,} \mathbb{P}(D=0 \mid X)> \alpha \\
Y(1) & \text{if \,\,} \mathbb{P}(D=0 \mid X) \leq \alpha
\end{cases}
\label{eq:trimming}
\end{equation}
where $\alpha$ is small and arbitrarily chosen. In the treatment effect literature, trimming is often performed for $\alpha=0.01$ or $\alpha=0.05$. However, there is no consensus on how to select $\alpha$. Putting (\ref{eq:trimming}) in (\ref{decompo_ref0}) leads to replacing $Y(0)$ by $Y(1)$ in the right component $\Delta_S^{0,1}$ for the individuals in group $D=1$ without (enough) counterparts only. In expectation, this amounts to removing those observations. Thus, trimming can be viewed as a specific extrapolation where no discrimination is assumed outside the support of the reference group. Another way to look at trimming is to consider that the population of interest has been modified and that we are no longer seeking to explain the gaps in the population for certain $X$ profiles. Empirically, \citet{strittmatter2021gender} show that the results of the decomposition are sensitive to the trimming used to have observationally comparable men and women, and thus, that limited overlap is a key issue to study gender wage gaps with large datasets.

Limited overlap issues can also potentially pose a problem for the estimator~(\ref{unexp:regression_g0_g1}) in a more indirect way. Indeed, if for certain regions of the support of $X$ the outcome is rarely observed for a given subgroup, the model estimate within that subgroup may be biased. The resulting extrapolation is therefore questionable. Linear regression models estimated using OLS and models estimated using a similar technique, such as Lasso with a few polynomials and interactions, certainly have the advantage of providing smoother estimates than other machine learning methods, but the fact remains that the analysis relies on an extrapolation whose quality cannot be properly verified due to a lack of observations. Nonetheless, this solution is the most standard one in classical decompositions. For example, if several individuals in group $D=1$ are old, which is rarely or never observed in group $D=0$, then $Y(0)$ will be predicted outside the support of $X \mid D=0$ on the basis of a linear relationship between the age and the observed outcome in group $D=0$. Hence, the actual reference outcome would be: 
\begin{equation}
\text{Extrapolation:} \quad Y_r = 
\begin{cases}
Y(0) & \text{if \,\,} \mathbb{P}(D=0 \mid X)> 0 \\
\tilde Y(0) & \text{if \,\,} \mathbb{P}(D=0 \mid X) =0
\end{cases}
\end{equation}
$\tilde Y(0)$ corresponds to a counterfactual for which there are no observed counterparts in group $D=0$ since $\mathbb{P}(D=0\mid X)=0$. Thus, it cannot be estimated on actual observations, in contrast to $Y(0)$ that is observed for some individuals in group $D=0$. The only way to make a guess about this wage would be to extrapolate $\tilde Y(0)$ based on an estimation of $Y(0)$. The extrapolation $\tilde Y(0)$ requires additional assumptions, for instance the use of a parametric regression being stable outside the support of the sample used for the estimation. In a non-parametric regression framework, this approach is more problematic, as extrapolation outside the support is generally infeasible.

An alternative approach to extrapolation and trimming in the presence of limited overlap is to focus on the decomposition based on the reference outcome $Y_e$ (equilibrium outcome). Indeed, considering the expressions of the corresponding estimators when choosing $w(X)=e(X)$ in Tables~\ref{tab:regression_estimators}, \ref{tab:ipws} and \ref{tab:aipw_estimators}, it can be seen that none of these five estimators involves dividing by the propensity score\footnote{Although the general formula was derived under strong assumptions, notably regarding overlap, we can start directly from the AIPW estimator based on $m(X)$ and show once again that it identifies the target value thanks to the equality $\Eb[Y^\text{obs}\mid X] = \Eb[Y_e \mid X]$, which holds by definition of $Y_e$ under Assumption~\ref{as:ignorability}.}.

Estimators based on Figure~\ref{fig:reference} illustrates the impact of the choice of the reference outcome on the difference in means decomposition with limited overlap. Two sets of outcome values are drawn from  linear regression models, for the advantaged group 1 and disadvantaged group 0, and the conditional probability model is a logit model.\footnote{$Y(1) = \mu_1(X) + \epsilon_1$ where $\mu_1(X)=0.3 + 0.42 \ X$ and $\epsilon_1 \sim \mathcal{N}(0,0.01)$ and $Y(0) = \mu_0(X) + \epsilon_0$ where $\mu_0(X)=0.2 + 0.2 \ X$ and $\epsilon_0 \sim \mathcal{N}(0,0.015)$ and $P(D=1|X)=1/(1+\exp(4-8 \ X))$}
This design illustrates the presence of very few pairs for high and low values of $X$. The unexplained part of the decomposition in~\eqref{decompo_refr} is based on $\mathbb{E}[Y_r\mid D=d]$, for $d\in\{0,1\}$. These quantities are obtained by averaging the conditional expectation $\mathbb{E}[Y_r\mid X]$, which is drawn with a black line in the figures, over the distribution of $X$ in group $d$.


    \begin{figure}[htpb!]
    \caption{The effect of reference outcome selection on the unexplained component under limited overlap}
    \label{fig:reference}
    \centering
    \begin{threeparttable}
    \includegraphics[width=\textwidth]{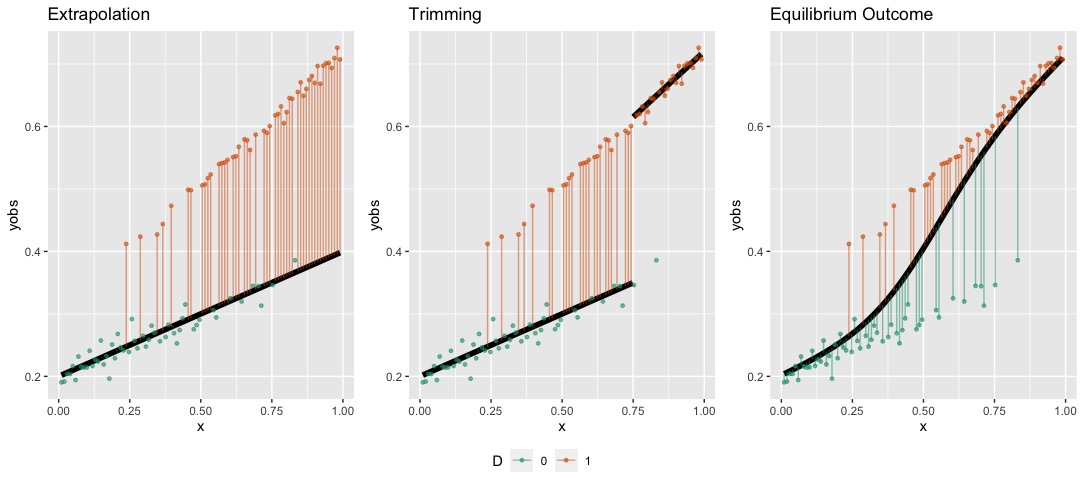}
        \begin{tablenotes}
            \footnotesize
            \item \textbf{Note}: On each graph, the reference outcome chosen is shown in black. This makes it possible to observe how each approach (extrapolation, trimming, and equilibrium outcome) deals with the limited overlap problem occurring for large $x$. (Left) In the case of extrapolation, although there are few observations in the group $0$ for $x \geq 0.75$, the linear model is extended. (Middle) For trimming, it is considered that there is not enough information from $x \geq 0.75$, and the distances between the two models are ignored for these individuals in group~1. (Right) Finally, in the case of the equilibrium outcome, all sample points (i.e., pooled sample observations) are compared with an average of the two potential outcomes models, weighted by the propensity score.
   \end{tablenotes}
    \end{threeparttable}
\end{figure}

\subsection{Does Neumark's approach tend to reduce the unexplained part? And why we should stop “correcting” the pooled model by adding a dummy variable to the regression.}

\

In this section, we examine the critiques by \cite{jann2008stata}, \cite{fortin2008gender}, and \cite{elder2010} of the reference proposed by \cite{neumark1988employers}. By proposing a framework that is both heterogeneous and nonlinear, we determine which of their arguments hold up in a general framework. Using the framework of potential outcomes, we are able to explain why the debate should not revolve around the issue of estimation bias in intermediate regressions but rather in terms of target values.

\subsubsection{Why we shouldn't always worry about bias in pooled regression}

Part of the criticism in \cite{jann2008stata} is misleading, although fairly straightforward. It begins with the true model for $Y^\text{obs}$ in the linear framework without heterogeneity:
\begin{align*}
    Y^\text{obs} = \beta_0 + \beta_1 X + \beta_2 D + \epsilon
\end{align*}

\cite{neumark1988employers} proposes estimating the following pooled regression:
\begin{align*}
    Y^\text{obs} = \tilde\beta_0 + \tilde \beta_1 X + \tilde \epsilon    
\end{align*}

This regression exposes us to an omitted variable bias (where $D$ is the omitted variable). This is indeed indisputable. The subtlety, however, is this: to some extent, it is not a problem in itself to estimate the coefficients $\beta_0$ and $\beta_1$ in a biased manner via pooled regression. In fact, our goal is not to estimate these coefficients correctly, but rather to correctly estimate $\mathbb{E}[Y^\text{obs}\mid X]=m(X)$ in order to plug this nuisance function into our estimators. Under the assumption of linearity of $m(X)$, the pooled regression proposed by \cite{neumark1988employers} works perfectly. Figure~\ref{fig:illustration_g2} illustrates this point. In this design, the group-specific regression functions $\mu_0(X)$ and $\mu_1(X)$ are strongly non-linear, so that a linear model within each group would be misspecified. The pooled regression function, however, is $m(X)=e(X)\,\mu_1(X)+(1-e(X))\,\mu_0(X)$, a mixture of the two group-specific functions weighted by the propensity score, and in this design this mixture happens to be linear in $X$ (yellow line). The linear regression of $Y^\text{obs}$ on $X$ proposed by \cite{neumark1988employers}, which omits $D$ (dashed line), then coincides with $m(X)$, even though it recovers neither $\mu_0(X)$ nor $\mu_1(X)$. Under the assumption of linearity of $m(X)$, this pooled regression therefore works perfectly. If linearity does not hold, machine learning methods can still estimate properly $m(X)$. We therefore believe that the argument regarding estimation bias is unfounded. However, the related, more empirical argument from \cite{fortin2008gender} still needs to be examined: the unexplained part targeted by the reference $Y_e$ tends to be smaller. This empirical argument is supported by the theoretical results of \cite{elder2010}. We examine it within our general framework.

\begin{figure}[tbp]
    \centering
    \includegraphics[width=0.6\linewidth]{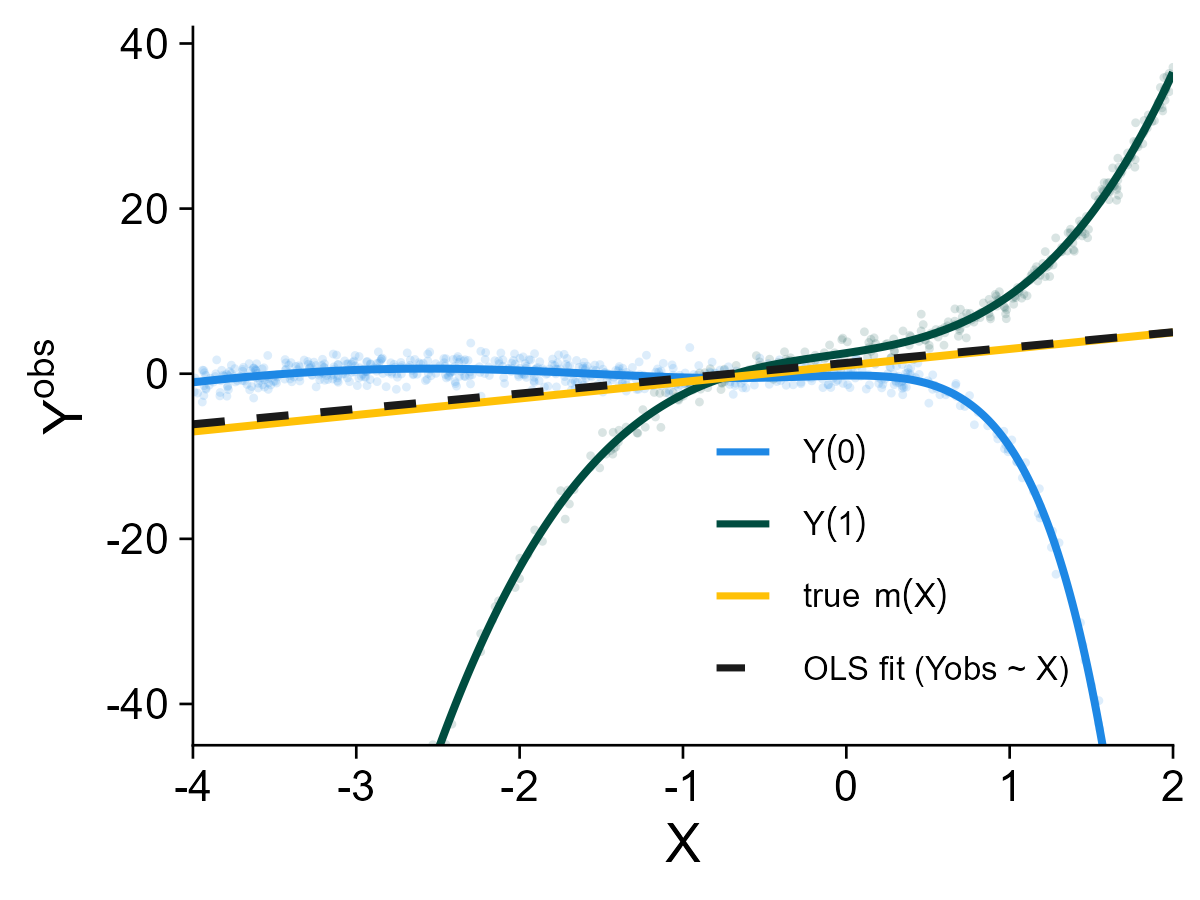}
    \caption{Estimating $m(X)$ does not require estimating the group-specific models}
    \label{fig:illustration_g2}
\end{figure}

\subsubsection{Under ignorability, $\delta_e$ is a shrinkage version of the coefficient in the regression of residuals on residuals, which makes it a target closer to $0$ than the standard targets in the absence of heterogeneity}
\label{sec:shrinkage_res}

\

The criticism in \cite{elder2010} can be placed within a more general, nonlinear framework. To do this, let us express the unexplained part associated with the choice of the reference $Y_e$ ($w(X) = e(X)$) in terms of the heterogeneity $\tau(X) = Y(1) - Y(0)$ and the shrinkage coefficient $R^2_D$ defined by:
$$R^2_D = \mathbb{E}[e(X) \mid D=1] - \mathbb{E}[e(X) \mid D=0] = \frac{V\big(\Eb[D\mid X]\big)}{V(D)}
$$

Following \cite{elder2010}, $R^2_D \in [0,1]$ is a coefficient that is closer to $1$ the better the covariates $X$ separates the two groups, and thus the more limited the overlap is. In our general framework, providing that the overlap assumption is verified, we can express $\delta_e$ in terms of $R^2_D$ as shown in Appendix~\ref{proof:delta_2_in_terms_of_R2}:

\begin{align}
\delta_e=(1-R^2_D) \times \frac{\mathbb{E}[\tau(X) e(X)(1-e(X))]}{\mathbb{E}[e(X)(1-e(X))]} \label{eq:delta2_as_shrinkage}
\end{align}

$R^2_D$ is responsible for the shrinkage phenomenon described in the literature. Indeed, in the absence of heterogeneity, $\tau(X)$ no longer varies with $X$. Assuming $\mathbb{E}[\tau(X)\mid X]=\mathbb{E}[\tau(X)]=\tau$ (hence, no heterogeneity), we have the following equality:
\begin{align}
\label{shrinkage_under_homogeneity}
\begin{cases}
\delta_0 = \delta_1 = \delta_{\pi} = \delta_{1/2} = \tau = \delta_w & \text{for any } w(\cdot) \text{ constant} \\[4pt]
\delta_e = \tau \, (1-R_D^2)
\end{cases}
\end{align}

See the proof in Appendix~\ref{proof:no_heterogeneity}. The observation is as follows: $\delta_e$ is a shrinkage version of all the other unexplained parts, and this shrinkage phenomenon is all the more significant when the overlap is limited.

An important lesson to be drawn from this is the following: when using $Y_e$ as a reference, one must under no circumstances choose the covariates $X$ based solely on their predictive power for $D$, without considering their economic relevance. For example, one should not choose an $X$ that would be a good predictor of gender but would have no relevance in determining the reference wage. A caricatured example might be the following: suppose that women tend to have long hair more often than men, such that hair length is a good predictor of gender and effectively separates the observations into two groups; in that case, adding this variable will have no impact on $\delta_1$ or $\delta_0$\footnote{Nonetheless, this limited overlap could increase the risk of incorrect extrapolation or reweighting due to extreme propensity scores for these references, as described above.}, but it will have an impact on $\delta_e$, which will approach $0$. We consider this criticism significant enough to recommend reporting $R^2_D$ in analyses when $Y_e$ is chosen as the reference variable, in order to easily detect this type of anomaly. 

However, the shrinkage property above, which makes $\delta_e$ a target closer to $0$ than the standard targets, relies on the absence of heterogeneity. Indeed, in the presence of heterogeneity, it is entirely possible that  $|\delta_e| \in [min(|\delta_0|,|\delta_1|), max(|\delta_0|,|\delta_1|)]$ or even $|\delta_e| > max(|\delta_0|,|\delta_1|)$. We therefore cannot reject the reference $Y_e$ on the grounds that it would yield a smaller unexplained part than the standard references $Y(0)$ or $Y(1)$. In other words, $Y_e$ as a reference in the presence of heterogeneity will not necessarily yield the unexplained part closest to zero, and may result in an unexplained part larger than the standard target values ($\delta_0$ and $\delta_1$). The following three simple counterexamples demonstrate this.

\begin{itemize}
    \item \textit{Counterexample A in presence of heterogeneity with $|\delta_e| \in (|\delta_0|, |\delta_1|)$}
Suppose that $X\sim Ber\left(\frac{3}{4}\right)$, $e(0)=\frac{1}{2}$, $e(1)=\frac{5}{6}$, $\tau(0)=1$, and $\tau(1)=0$. The fact that $\tau(0)\neq\tau(1)$ implies that the effects are heterogeneous. Appendix~\ref{proof:counterexample1} shows that in this specific case $\delta_0 = \frac{1}{6}$, $\delta_e = \frac{1}{3}$, $\delta_1=\frac{1}{2}$, which implies that $|\delta_0| < |\delta_e| < |\delta_1|$.


\item \textit{Counterexample B in presence of heterogeneity with $|\delta_e| > max(|\delta_0|, |\delta_1|)$} 

In the presence of heterogeneity, $\delta_e$ may not only fail to be the unexplained component closest to zero, compared to $\delta_0$ and $\delta_1$, but can even lie farther from zero than both $\delta_0$ and $\delta_1$. A simple example shows that this possibility exists. Let us assume that the support of $X$, $\mathcal{X} = \{0,1,2\}$, with $P(X=0)=0.2$, $P(X=1)=0.2$, $P(X=2)=0.6$. Moreover, let us assume that $e(0)=0.2$, $e(1)=0.8$, $e(2)=0.5$. In addition, $\tau(0)=-1$, $\tau(1)=-1$, $\tau(2)=1$. Appendix~\ref{proof:counterexample2} shows that in this specific case $\delta_0 = 0.2$, $\delta_1=0.2$, $\delta_e = 0.344$, which implies that simultaneously $|\delta_e| > |\delta_0|$ and $|\delta_e| > |\delta_1|$.
\item \textit{Counterexample C in presence of heterogeneity with $|\delta_e| > max(|\delta_0|, |\delta_1|)$} and a high $R^2_D$

Consider this last example. $\mathcal{X}=\{0,1,2,3\}$ associated to $P(X=0)=0.29$, $P(X=1)=0.32$, $P(X=2)=0.11$ and $P(X=3)=0.28$. Heterogeneity is defined by $\tau(0)=-1$, $\tau(1)=-1$, $\tau(2)=2$ and $\tau(3)=2$. The propensity scores are $e(0)=0.88$, $e(1)=0.03$, $e(2)=0.95$ and $e(3)=0.27$. In this case, $\delta_0 \approx 0.214$, $\delta_1 = 0.134$, $\delta_e\approx 0.328$ and $R^2 \approx 0.594$. This example shows that even when shrinkage is strong, $\delta_0$ and $\delta_1$ may be closer to zero than $\delta_e$.

\end{itemize}

That said, the equality~(\ref{eq:delta2_as_shrinkage}) still holds. To better understand what $\delta_e$ is as a shrinkage version in the presence of heterogeneity, let us consider the regression of residuals on residuals. Let $\widetilde D = D - \mathbb{E}[D\mid X] = D-e(X)$ and $\widetilde Y^\text{obs}=Y^\text{obs} - \mathbb{E}[Y^\text{obs}\mid X]$. Consider the coefficient of the linear regression of $\widetilde Y^\text{obs}$ on $\widetilde D$—that is, the regression of residuals on residuals—obtained by ordinary least squares:
\begin{align*}
    \delta_\text{res} = \frac{\mathbb{E}[\tilde{D} \tilde{Y}^{\text{obs}}]}{\mathbb{E}[\tilde{D}^2]}
\end{align*}

In Appendix~\ref{proof:residuals_on_residuals}, we show, under the assumption of ignorability (overlap being already assumed), that the following equality holds:
\begin{align}
    \delta_\text{res} = \frac{\mathbb{E}[\tau(X) e(X)(1-e(X))]}{\mathbb{E}[e(X)(1-e(X))]} \label{eq:delta_res_f_of_tau}
\end{align}

This gives us the equality:
\begin{align*}
    \delta_e = (1-R^2_D) \times \delta_\text{res}
\end{align*}

Thus, $\delta_e$ remains a shrinkage version of an object, $\delta_\text{res}$. Figure~\ref{fig:illustration_shrinkage} summarizes these results. Case 1 corresponds to a situation without heterogeneity: $\delta_e$ is less than or equal to all other references and to $\delta_\text{res}$. In Case 2 of Figure~\ref{fig:illustration_shrinkage}, we are dealing with heterogeneity. This corresponds to the counterexamples A, B and C mentioned above. $\delta_e$ remains a shrinkage version of $\delta_\text{res}$ but no longer of the unexplained parts, $\delta_0$ and $\delta_1$.

We now also have a better understanding of what the ``corrections'' proposed by \cite{jann2008stata} and \cite{fortin2008gender} essentially entail. These authors proposed adding a group dummy variable to the pooled regressions, which directly addresses the (unfounded) criticism of an omitted variable in the proposal by \cite {neumark1988employers}. The unexplained part was then associated with the coefficient in front of the group dummy. In the case without heterogeneity (and with a linear outcome model), their target value when adding a dummy in the pooled regression coincides with $\tau$, hence with all the unexplained parts in the first line of the system of equations~(\ref{shrinkage_under_homogeneity}), and with $\delta_\text{res}$. In the heterogeneous case, this first line of the system of equations no longer holds. However, the target value still equals $\delta_\text{res}$, assuming that $e(X)$ is linear in $X$. In other words, adding the group dummy undoes the shrinkage of $\delta_e$: the coefficient on the dummy targets $\delta_\text{res}=\delta_e/(1-R^2_D)$, that is, $\delta_e$ without the factor $(1-R^2_D)$. This restores $\tau$ in the absence of heterogeneity, but under heterogeneity $\delta_\text{res}$ may be far from both $\delta_0$ and $\delta_1$ (see counterexample C).\footnote{In counterexample C, $\delta_\text{res}\approx 0.806$ while $\delta_0 \approx 0.214$ and $\delta_1 \approx 0.134$.} Yet $\delta_\text{res}$ may not be a relevant target.

Indeed, $\delta_\text{res}$, as shown in equation~(\ref{eq:delta_res_f_of_tau}), is associated with a very specific weighting of the potential outcomes $Y(1)$ and $Y(0)$. Let us assume that $\delta_\text{res}$ is associated with a reference outcome that is a weighted average of $Y(1)$ and $Y(0)$. Solving for the weights yields:

$$
w^\star(X) \;=\; \frac{(1-\pi)\,e(X)\,\big[\pi\big(1-e(X)\big) - C\big]}{C\,\big(\pi - e(X)\big)},
\qquad C = \Eb\big[e(X)\big(1-e(X)\big)\big]$$

provided that $e(X) \neq \pi$, as shown in Appendix~\ref{proof:residuals_on_residuals_weights}.

Table~\ref{tab:wstar_counterexamples} displays the weights $w^\star(\cdot)$ creating a mixture of $Y(1)$ and $Y(0)$ that yields $\frac{\mathbb{E}[\tau(X) e(X)(1-e(X))]}{\mathbb{E}[e(X)(1-e(X))]} $ (and thus $\delta_\text{res}$ under ignorability) as an unexplained part. In some cases, these weights may be undefined because of the division by $\pi - e(X)$ (case B), or may lie outside $[0,1]$ (case C). In these cases, $\delta_\text{res}$ therefore represents an unexplained part that does not correspond to a choice of a single reference outcome within the family $\mathcal{F}$ and is thus very difficult to interpret economically. This is why we recommend, outside of a purely linear framework without heterogeneity, avoiding the “correction” proposed by \cite{jann2008stata} and \cite{fortin2008gender}.

\begin{table}[h!]
\centering
\caption{Implicit weight $w^\star(X)$ such that the reference outcome $Y_{w^\star} = w^\star(X)\,Y(1) + (1-w^\star(X))\,Y(0)$ yields $\delta_{\mathrm{res}}$ as an unexplained part, across the three counterexamples}
\label{tab:wstar_counterexamples}
\begin{threeparttable}
\begin{tabular}{clcc}
\toprule
Counterexample & $X$ & $w^\star(X)$ & $\delta_{\mathrm{res}}$ associated to an outcome in $\mathcal{F}$? \\
\midrule
\multirow{2}{*}{A}
 & $0$ & $0.625$ & \multirow{2}{*}{Yes} \\
 & $1$ & $0.625$ & \\
\midrule
\multirow{3}{*}{B}
 & $0$ & $0.290$ & \multirow{3}{*}{Undefined} \\
 & $1$ & $0.710$ & \\
 & $2$ & \textit{undefined} ($e(2)=\pi$) & \\
\midrule
\multirow{4}{*}{C}
 & $0$ & $0.525$ & \multirow{4}{*}{No} \\
 & $1$ & $0.133$ & \\
 & $2$ & $0.813$ & \\
 & $3$ & $1.917$ & \\
\bottomrule
\end{tabular}
\begin{tablenotes}
\footnotesize
\item \textbf{Note}: $Y_{w^\star}$ belongs to the family $\mathcal{F}$ only when $w^\star(X)\in[0,1]$ for all $X$; this is not the case in counterexample~C, where $w^\star(3)>1$. In counterexample~B, the weights $w^\star(\cdot)$ are not defined. 
\end{tablenotes}
\end{threeparttable}
\end{table}
It should be noted that it is not surprising that, in the homogeneous case, this curious reference outcome, which lies outside the family $\mathcal{F}$, yields exactly $\tau$. In fact, in the homogeneous case, many reference outcomes can yield $\tau$ without belonging to $\mathcal{F}$ and without having much economic relevance. The first line of the system of equations is indeed always satisfied for a constant $w$, and the proof in Appendix~\ref{proof:no_heterogeneity} holds for any constant $k$, without assuming that $k \in [0,1]$. Thus, in the context of the gender pay gap, if we were to assume that the reference wage should be 10 times that of men minus 9 times that of women, that is $k = 10$, then we would obtain $\tau$ as the unexplained part. However, this does not necessarily make it a good choice for a reference outcome. We cannot, therefore, reject the proposal by \cite{neumark1988employers} on the grounds that the associated unexplained part ``does not conform to the norm'' in the homogeneous case\footnote{This discussion on the interpretation of target values outside homogeneous cases according to different estimation techniques is philosophically reminiscent of the discussion in a completely different context by \cite{de2020two}. They emphasize that, in the presence of heterogeneity, certain estimation techniques lead empiricists to estimate target values associated with unintended weightings of the average treatment effects. Our discussion thus demonstrates once again the strength of the potential outcomes framework for interpreting what estimation techniques actually do.}.

\subsection{The need for an economic theory of the reference outcome choice} 
\label{sec:need_for_an_econ_theory}

The reasoning in \cite{neumark1988employers} is unique in that it links the choice of the reference outcome to economic considerations in the context of the gender pay gap. Building on the taste-based discrimination framework of \citet{becker1957economics} and \citet{arrow1972some}, \citet{neumark1988employers} shows that the Kitagawa-Oaxaca-Blinder decomposition can be derived from a theoretical model of employers' discriminatory tastes. Setting $Y_r=Y(1)$ reflects the idea of ``pure'' discrimination against women, while setting $Y_r=Y(0)$ reflects the idea of ``pure'' nepotism toward men, that is, a willingness to hire men even at relatively higher wages. His proposed reference outcome is based on an economic model in which the employers can be both discriminatory against females and nepotistic toward males. The reference outcome $Y_e$ that we propose is a direct generalization of the reference outcome proposed by \cite{neumark1988employers}. For an individual working in a given type of labor $l$,\footnote{For \citet{neumark1988employers}, ``types of labor'' are job categories, illustrated by the distinction between skilled versus unskilled jobs. Alternative definitions could be based on other criteria, such as employment sectors.}  the reference outcome was initially defined as the weighted average: 
\begin{equation}
Y_\text{Neumark}(2)=w_l\, Y(1) + (1-w_l)\, Y(0)
\label{eq:y_neumark}
\end{equation}
where $w_l$ and $(1-w_l)$ are the proportions of men and women working in labor type $l$. Here, the non-discriminatory wage is sensitive to the gender composition of each type of labor. This avoids assuming either that the non-discriminatory wage corresponds to men’s wages (which would imply that employers would forgo their profits by systematically paying the higher wage), or that it corresponds to women’s wages, which would imply that the wage premium enjoyed by men stems from pure favoritism and not from labor market dynamics. \cite{neumark1988employers} brings the idea of market equilibrium back to the center of the discussion. Thus, a major advantage of this choice of reference outcome is that it is based on a clear theory of discrimination on the labor market and takes equilibrium effects into account, as underlined by \citet{sloczynski2020average}.

That said, one can still criticize the economic foundations of the generalized reference outcome $Y_e$ when there are numerous covariates. Indeed, the more covariates we add to the model, the more we segment the market into micromarkets within which equilibrium mechanisms play out. This implies that workers are assumed not to be mobile between micromarkets and would not destabilize the equilibrium if wage premiums were to disappear. A radical alternative is to assume perfect labor mobility, no entry into or exit from the labor market, and a rebalancing of wages in the counterfactual world. These assumptions would lead to the choice of the reference outcome $Y_{\pi}$, according to which the entire wage bill is distributed equally between the two groups. The most realistic reference outcome is likely neither $Y_e$ nor $Y_{\pi}$, but rather a reference outcome derived from market equilibria based on less extreme assumptions. We therefore invite empirical researchers to reflect on their choice of counterfactual, but also to consider new reference scenarios that have not yet been proposed in the literature.

\section{Conclusion}

Situations involving limited overlap remain a challenge for inequality decompositions. This problem already existed in standard decompositions based on linear regressions, as they forced us to extrapolate into regions of $X$ for which we did not have enough observations. Since these methods did not involve propensity score reweighting, the problem went largely unnoticed, but the trade-off was a very strong linearity assumption. Large datasets allow for greater flexibility in the models, but they also risk making the limited overlap problem even more pronounced, for two reasons. First, machine learning plug-in methods, such as double machine learning, rely on the AIPW estimator and thus on propensity score reweighting. Second, large databases allow for better classification of individuals into the two categories and therefore potentially produce more extreme propensity scores.

That said, beyond this general issue, the challenges vary depending on the choice of reference outcome and the estimation method used. This article explains exactly why. Estimators based on subgroup regression rely on extrapolation, which can be compromised by limited overlap due to poor prediction in certain regions of the support of $X$. Estimators based on pooled regression suffer less directly from this problem, since individuals from both groups are retained, but they rely on division by the propensity score, which makes the estimate unstable in cases of limited overlap. The IPW and AIPW estimators are also affected by these divisions by the propensity score, except in the case $w(X)=e(X)$, that is, when the reference outcome $Y_e$ is used, as inspired by the work of \citet{neumark1988employers}. However, in this case, the estimator is sensitive to the inclusion of variables $X$ that strongly predict $D$, due to a shrinkage effect already highlighted by \citet{elder2010}, whose results we generalize here. In particular, this leads us to recommend that practitioners take care not to include variables $X$ that predict $D$ without being related to the observed outcome, if they use the reference from \citet{neumark1988employers}.

The choice of the reference outcome remains a theoretical consideration that must be situated within the specific empirical context of decompositions. Thus, this paper does not settle it, but explores all of its implications instead. Generalizing through the framework of potential outcomes allows for a better understanding of the debates that were initiated within the linear framework, and for a discussion of the cases in which the criticisms leveled against \citet{neumark1988employers} hold in a general case or not.

Finally, this article highlights an original result on Neyman orthogonality, which had not been discussed in the literature until now. By presenting two distinct AIPW estimators—one based on subgroup outcome models ($\mu_0(X)$, $\mu_1(X)$) and the other on the observed outcome model $m(X)$—we were able to construct two distinct score functions. The first was presented in a specific case (that of ATT estimation) by \citet{chernozhukov2018double}; we present it here in a more general case, allowing for the use of different reference outcomes. The second score function is an innovation specific to this article. We highlight a surprising result: the two score functions are Neyman-orthogonal under mutually exclusive conditions on $w(X)$. The first score function satisfies Neyman orthogonality for any $w(X)$ that does not depend on the nuisance functions (typically, a constant such as $w=0,1,\pi$, or $1/2$), but does not satisfy it for $w(X)=e(X)$, that is, in the context of the generalization of the approach in \citet{neumark1988employers}. Conversely, the second score function satisfies Neyman orthogonality for $w(X)=e(X)$, but does not satisfy it for the other standard reference outcomes. For readers familiar with the literature on treatment effects, this observation serves as a reminder that simply constructing an AIPW estimator is not sufficient to guarantee the Neyman orthogonality of the score. This result also illustrates that the initial idea proposed by \citet{neumark1988employers}, which involved performing a pooled regression of the observed outcome on the predictors while omitting the group variable, held potential that extended beyond the linear framework with categorical features.

\newpage

\begin{algorithm}[t!]
\caption{Machine Learning Estimation of the Unexplained Part for Constant $w(X)$}\label{algo:DR-constant}
\begin{algorithmic}[1]
\Require
\Statex Data $(Y_i^{obs}, D_i, X_i)$ for $i=1,\ldots,N$, where $Y_i^{obs}$ is the observed outcome, $D_i$ is the group, and $X_i$ is a vector of covariates.
\Statex $w \in [0,1]$: a constant choice of reference outcome weight (e.g. $w=0$, $w=1$, $w=\pi$, or $w=\tfrac{1}{2}$).
\Statex $K$: number of repetitions of the cross-fitting procedure.
\Ensure Estimated unexplained part $\widehat\delta_w$ and its standard error $\text{se}(\widehat\delta_w)$.
\For{$k=1$ to $K$}
\State Take a random subsample $I_k$ of the full sample $I$, and define the auxiliary sample $I_k^c = I \backslash I_k$, which includes all observations that are not in $I_k$.
\State \textit{First stage}. Estimate the nuisance functions $\mu_0(\cdot)$, $\mu_1(\cdot)$, and $e(\cdot)$ from the auxiliary subsample $I_k^c$, denoted $\hat\mu_{0,k}^c(\cdot)$, $\hat\mu_{1,k}^c(\cdot)$, and $\hat e_k^c(\cdot)$, with standard machine learning methods, e.g. lasso, ridge, random forests, boosting, or neural networks. If $w(X)=0$, only an estimator of $\mu_0(X)$ is required. If $w(X) =1$, only an estimator of $\mu_1(X)$ is required.
\State \textit{Second stage}. Compute $\widehat\delta_{w,k}$, the estimator of the unexplained part, from the subsample $I_k$, based on the AIPW score $\psi_1$, equation~\eqref{score:aipw_g0g1}, using $\hat\mu_{0,k}^c(\cdot)$, $\hat\mu_{1,k}^c(\cdot)$, and $\hat e_k^c(\cdot)$.
\EndFor
\State Compute the final estimator as the average of the $K$ estimates: $\widehat\delta_w = \dfrac{1}{K}\sum_{k=1}^K \widehat\delta_{w,k}$
\State Compute the estimated variance $\hat\sigma^2 = \dfrac{1}{K}\sum_{k=1}^K \dfrac{1}{n_k}\sum_{i\in I_k} \psi_1^2\big(W_i;\widehat\delta_{w,k},\hat\eta_k\big)$, and the standard error $\text{se}(\widehat\delta_w) = \hat\sigma/\sqrt{n}$.
\end{algorithmic}
\end{algorithm}

\begin{algorithm}[t!]
\caption{Machine Learning Estimation of the Unexplained Part for $w(X)=e(X)$}\label{algo:DR-neumark}
\begin{algorithmic}[1]
\Require
\Statex Data $(Y_i^{obs}, D_i, X_i)$ for $i=1,\ldots,N$, where $Y_i^{obs}$ is the observed outcome, $D_i$ is the group, and $X_i$ is a vector of covariates.
\Statex $K$: number of repetitions of the cross-fitting procedure.
\Ensure Estimated unexplained part $\widehat\delta_e$ and its standard error $\text{se}(\widehat\delta_e)$.
\For{$k=1$ to $K$}
\State Take a random subsample $I_k$ of the full sample $I$, and define the auxiliary sample $I_k^c = I \backslash I_k$, which includes all observations that are not in $I_k$.
\State \textit{First stage}. Estimate the nuisance functions $m(\cdot)$ and $e(\cdot)$ from the auxiliary subsample $I_k^c$, denoted $\hat m_k^c(\cdot)$ and $\hat e_k^c(\cdot)$, with standard machine learning methods, e.g. lasso, ridge, random forests, boosting, or neural networks.
\State \textit{Second stage}. Compute $\widehat\delta_{e,k}$, the estimator of the unexplained part, from the subsample $I_k$, based on the AIPW score $\psi_2$, equation~\eqref{score:aipw_g2}, using $\hat m_k^c(\cdot)$ and $\hat e_k^c(\cdot)$.
\EndFor
\State Compute the final estimator as the average of the $K$ estimates: $\widehat\delta_e = \dfrac{1}{K}\sum_{k=1}^K \widehat\delta_{e,k}$
\State Compute the estimated variance $\hat\sigma^2 = \dfrac{1}{K}\sum_{k=1}^K \dfrac{1}{n_k}\sum_{i\in I_k} \psi_2^2\big(W_i;\widehat\delta_{e,k},\hat\eta_k\big)$, and the standard error $\text{se}(\widehat\delta_e) = \hat\sigma/\sqrt{n}$.
\end{algorithmic}
\end{algorithm}

\begin{landscape}
\begin{figure}
    \centering
\includegraphics[width=\linewidth]{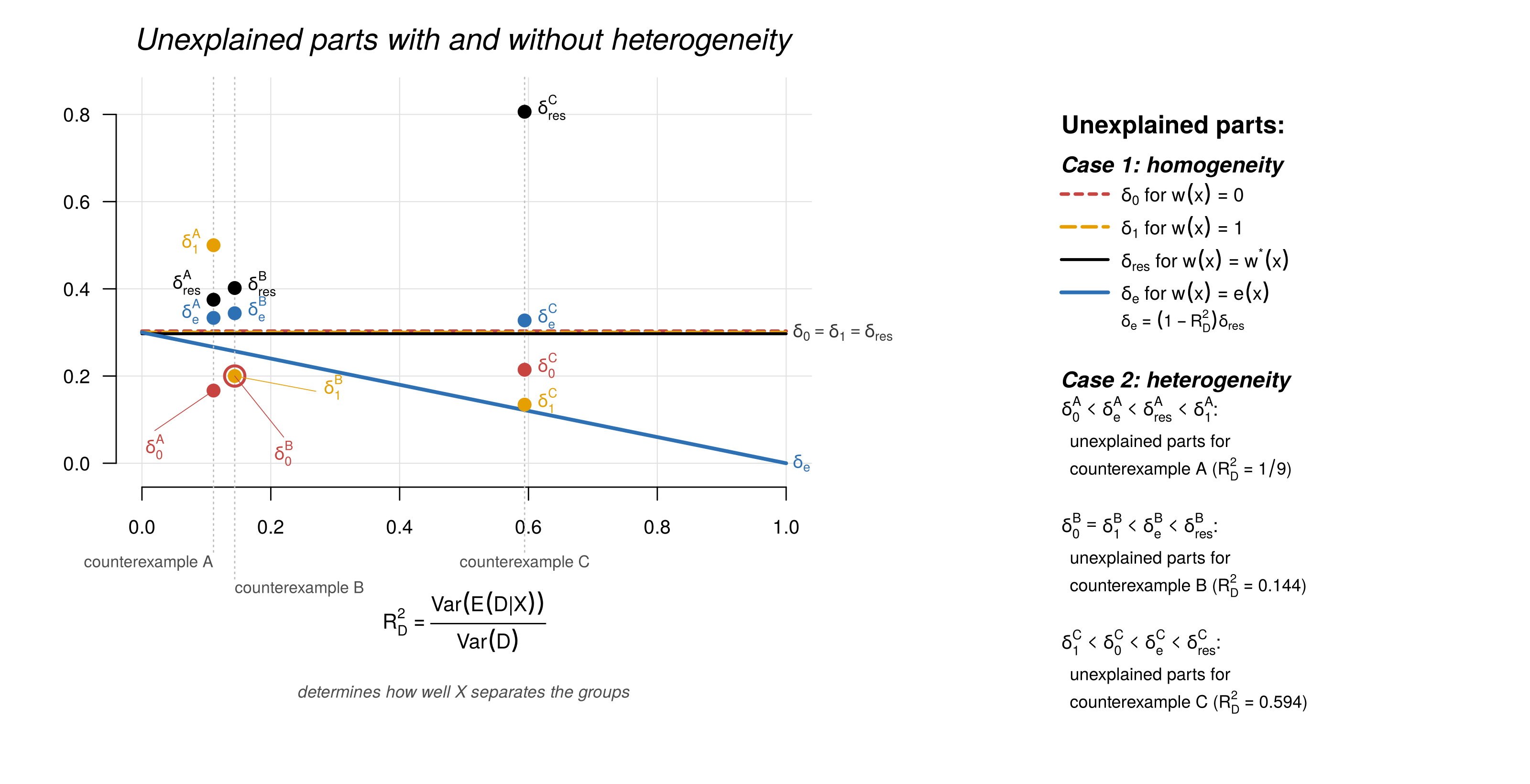}
    \caption{$\delta_e$ as a shrinkage version of the unexplained part in presence and in absence of heterogeneity}
\label{fig:illustration_shrinkage}
\end{figure}
\end{landscape}

\newpage
\appendix
\renewcommand{\thesection}{\Alph{section}}

\begin{landscape}

\section{Estimators of the unexplained parts for standard reference outcomes}

\begin{table}[h!]
\centering
\caption{Estimators of the unexplained part $\delta_w$ using the regression approach}
\label{tab:regression_estimators}

\noindent\textit{Notation: $e=e(X)=P(D=1\mid X)$, $\pi=P(D=1)$, $Y=\Yobs$, $\mu_d=\mu_d(X)$, $m=m(X)$. All entries are of the form $\delta_w=E[\,\cdot\,]$. Each cell gives the argument of the expectation.}

\vspace{1.2em}

\renewcommand{\arraystretch}{2.6}
\setlength{\tabcolsep}{6pt}
\small
\begin{tabular}{c c c c c}
\toprule
\multirow{2}{*}{\shortstack{Reference\\outcome}} & \multirow{2}{*}{$w(X)$} & \multirow{2}{*}{Source} & \multicolumn{2}{c}{\textbf{Regression approach}} \\
\cmidrule(lr){4-5}
 & & & via $\mu_1(X),\,\mu_0(X)$ & via $m(X)$ \\
\midrule
$Y(0)$ & $0$ & \shortstack{standard\\decomposition (ATT)}
& $\dfrac{D(Y-\mu_0)}{\pi}$
& $\dfrac{D(Y-m)}{\pi(1-e)}$ \\
$Y(1)$ & $1$ & \shortstack{standard\\decomposition (ATU)}
& $-\dfrac{(1-D)(Y-\mu_1)}{1-\pi}$
& $-\dfrac{(1-D)(Y-m)}{(1-\pi)\,e}$ \\
$Y_e$ & $e(X)$ & \citet{neumark1988employers}
& $\dfrac{(1-e)D(Y-\mu_0)}{\pi}-\dfrac{e(1-D)(Y-\mu_1)}{1-\pi}$
& $\dfrac{D(Y-m)}{\pi}-\dfrac{(1-D)(Y-m)}{1-\pi}$ \\
$Y_{\pi}$ & $\pi$ & \citet{cotton1988decomposition}
& $\dfrac{(1-\pi)D(Y-\mu_0)}{\pi}-\dfrac{\pi(1-D)(Y-\mu_1)}{1-\pi}$
& $\dfrac{(1-\pi)D(Y-m)}{\pi(1-e)}-\dfrac{\pi(1-D)(Y-m)}{(1-\pi)\,e}$ \\
$Y_{1/2}$ & $\tfrac{1}{2}$ & \citet{reimers1983labor}
& $\dfrac{D(Y-\mu_0)}{2\pi}-\dfrac{(1-D)(Y-\mu_1)}{2(1-\pi)}$
& $\dfrac{D(Y-m)}{2\pi(1-e)}-\dfrac{(1-D)(Y-m)}{2(1-\pi)\,e}$ \\
\bottomrule
\end{tabular}
\end{table}

\end{landscape}

\begin{table}[htbp]
\centering
\caption{Estimators of the unexplained part $\delta_w$ using the IPW approach}
\label{tab:ipws}

\noindent\textit{Notation: $\pi=P(D=1)$, $Y=\Yobs$, $e=e(X)$. All entries are of the form $\delta_w=E[\,\cdot\,]$. Each cell gives the argument of the expectation.}

\vspace{1.2em}

\renewcommand{\arraystretch}{2.6}
\setlength{\tabcolsep}{6pt}
\footnotesize
\begin{tabular}{c c c c}
\toprule
\shortstack{Reference\\outcome} & $w(X)$ & Source & IPW \\
\midrule
$Y(0)$ & $0$ & \shortstack{standard\\decomposition (ATT)}
& $\dfrac{(D-e)\,Y}{\pi\,(1-e)}$ \\
$Y(1)$ & $1$ & \shortstack{standard\\decomposition (ATU)}
& $\dfrac{(D-e)\,Y}{(1-\pi)\,e}$ \\
$Y_e$ & $e$ & \citet{neumark1988employers}
& $\dfrac{(D-e)\,Y}{\pi\,(1-\pi)}$ \\
$Y_{\pi}$ & $\pi$ & \citet{cotton1988decomposition}
& $\left(\dfrac{1-\pi}{\pi\,(1-e)}+\dfrac{\pi}{(1-\pi)\,e}\right)(D-e)\,Y$ \\
$Y_{1/2}$ & $\tfrac{1}{2}$ & \citet{reimers1983labor}
& $\left(\dfrac{1}{2\pi\,(1-e)}+\dfrac{1}{2(1-\pi)\,e}\right)(D-e)\,Y$ \\
\bottomrule
\end{tabular}
\end{table}

\begin{landscape}
\begin{table}[p]
\centering
\caption{Estimators of the unexplained part $\delta_w$ using the AIPW approach}
\label{tab:aipw_estimators}

\noindent\textit{Notation: $\pi=P(D=1)$, $Y=\Yobs$, $\mu_d=\mu_d(X)$, $m=m(X)$, $e=e(X)$. All entries are of the form $\delta_w=E[\,\cdot\,]$. \\ Each cell gives the argument of the expectation.}
\vspace{1.2em}
\renewcommand{\arraystretch}{2.8}
\setlength{\tabcolsep}{6pt}
\footnotesize
\begin{tabular}{c c c c c}
\toprule
\multirow{2}{*}{\shortstack{Reference\\outcome}} & \multirow{2}{*}{$w(X)$} & \multirow{2}{*}{Source} & \multicolumn{2}{c}{\textbf{AIPW approach}} \\
\cmidrule(lr){4-5}
 & & & via $\mu_1(X),\,\mu_0(X)$ & via $m(X)$ \\
\midrule
$Y(0)$ & $0$ & \shortstack{standard\\decomposition (ATT)}
& $\dfrac{1}{\pi}(Y-\mu_0)\,\dfrac{D-e}{1-e}$
& $\dfrac{1}{(1-e)\,\pi}(Y-m)(D-e)$ \\
$Y(1)$ & $1$ & \shortstack{standard\\decomposition (ATU)}
& $\dfrac{1}{1-\pi}(Y-\mu_1)\,\dfrac{D-e}{e}$
& $\dfrac{1}{e\,(1-\pi)}(Y-m)(D-e)$ \\
$Y_e$ & $e$ & \citet{neumark1988employers}
& $\dfrac{1}{\pi}(Y-\mu_0)(D-e) + \dfrac{1}{1-\pi}(Y-\mu_1)(D-e)$
& $\left(\dfrac{1}{\pi}+\dfrac{1}{1-\pi}\right)(Y-m)(D-e)$ \\
$Y_{\pi}$ & $\pi$ & \citet{cotton1988decomposition}
& $\dfrac{1-\pi}{\pi}(Y-\mu_0)\,\dfrac{D-e}{1-e} + \dfrac{\pi}{1-\pi}(Y-\mu_1)\,\dfrac{D-e}{e}$
& $\left(\dfrac{1-\pi}{(1-e)\,\pi}+\dfrac{\pi}{e\,(1-\pi)}\right)(Y-m)(D-e)$ \\
$Y_{1/2}$ & $\tfrac{1}{2}$ & \citet{reimers1983labor}
& $\dfrac{1}{2\pi}(Y-\mu_0)\,\dfrac{D-e}{1-e} + \dfrac{1}{2(1-\pi)}(Y-\mu_1)\,\dfrac{D-e}{e}$
& $\left(\dfrac{1}{2(1-e)\,\pi}+\dfrac{1}{2\,e\,(1-\pi)}\right)(Y-m)(D-e)$ \\
\bottomrule
\end{tabular}
\end{table}
\end{landscape}

\section{Unexplained part for reference outcomes in $\mathcal{F}$}
\label{appendix:inexp-Fw-1}

Let $Y_w$ be a choice of reference outcome associated with $w(\cdot)$ within the family $\mathcal{F}$. The unexplained part associated with this choice of reference outcome is written:
\begin{align*}
\delta_w &= E[Y(1)\mid D=1] - E[Y_w\mid D=1] - \big(E[Y(0)\mid D=0] - E[Y_w\mid D=0]\big)\\
&= E[Y(1)\mid D=1] - E[w(X)Y(1) + (1-w(X))Y(0)\mid D=1]\\
&\quad - \big(E[Y(0)\mid D=0] - E[w(X)Y(1) + (1-w(X))Y(0)\mid D=0]\big)\\
&= E[(1-w(X))(Y(1) - Y(0))\mid D=1] + E[w(X)(Y(1) - Y(0))\mid D=0]
\end{align*}
using the notation $\tau(X) = Y(1) - Y(0)$:
\begin{align*}
&= E[(1-w(X))\,\tau(X)\mid D=1] + E[w(X)\,\tau(X)\mid D=0]\\
&= E\!\left[\frac{(1-w(X))\,\tau(X)}{\pi}\,D\right] + E\!\left[\frac{w(X)\,\tau(X)}{1-\pi}\,(1-D)\right]
\end{align*}
Hence:
\[
\delta_w = E\!\left[\frac{(1-\pi)(1-w(X))\,\tau(X)\,D + \pi\, w(X)\,\tau(X)\,(1-D)}{\pi(1-\pi)}\right]
\]

\section{Representation of $\delta_w$ in terms of $\mu_1(X)$ and $\mu_0(X)$}
\label{appendix:inexp-Fw-3}

Starting from~\eqref{dec:unexpl:as_sum_obs_and_count}:
\begin{align*}
\delta_w = & E\!\left[\frac{(1-\pi)(1-w(X))\,Y^\text{obs}D - \pi\, w(X)\,Y^\text{obs}(1-D)}{\pi(1-\pi)}\right] \\ & - E\!\left[\frac{(1-\pi)(1-w(X))\,Y(0)D - \pi\, w(X)\,Y(1)(1-D)}{\pi(1-\pi)}\right]
\end{align*}

Since $Y^\text{obs}D = \big(Y(1)D + Y(0)(1-D)\big)D = Y(1)D^2 = Y(1)D$ and $Y^\text{obs}(1-D) = \big(Y(1)D + Y(0)(1-D)\big)(1-D) = Y(0)(1-D)^2 = Y(0)(1-D)$, the counterfactual terms in the second expectation, $Y(0)D$ and $Y(1)(1-D)$, remain the only unobservable quantities.

By CIA, conditioning on $X$:
\[
E[Y(0)\,D\mid X] = E[Y(0)\mid X]\,E[D\mid X] = \mu_0(X)\,e(X),
\]
where the last equality uses Assumption~3, $\mu_0(X) = E[Y(0)\mid X,D=0] = E[Y(0)\mid X]$ (by CIA). Since $E[\mu_0(X)\,D\mid X] = \mu_0(X)\,E[D\mid X] = \mu_0(X)\,e(X)$ as well, we obtain
\[
E[Y(0)\,D\mid X] = E[\mu_0(X)\,D\mid X], \qquad \text{hence} \qquad E[Y(0)\,D] = E[\mu_0(X)\,D].
\]

Symmetrically, by Assumption~2, $\mu_1(X) = E[Y(1)\mid X,D=1] = E[Y(1)\mid X]$ (by CIA), so
\[
E[Y(1)\,(1-D)] = E[\mu_1(X)\,(1-D)].
\]

Substituting $Y(0)D$ with $\mu_0(X)D$ and $Y(1)(1-D)$ with $\mu_1(X)(1-D)$ in the counterfactual term:
\begin{align*}
\delta_w = & E\!\left[\frac{(1-\pi)(1-w(X))\,Y^\text{obs}D - \pi\, w(X)\,Y^\text{obs}(1-D)}{\pi(1-\pi)}\right] \\ & - E\!\left[\frac{(1-\pi)(1-w(X))\,\mu_0(X)D - \pi\, w(X)\,\mu_1(X)(1-D)}{\pi(1-\pi)}\right]
\end{align*}

Hence:
\[
\delta_w = E\!\left[\frac{(1-\pi)(1-w(X))\,(Y^\text{obs}-\mu_0(X))\,D - \pi\, w(X)\,(Y^\text{obs}-\mu_1(X))\,(1-D)}{\pi(1-\pi)}\right]
\]

that is,
\[
\delta_w = E\!\left[\frac{(1-w(X))\,D\,(Y^\text{obs}-\mu_0(X))}{\pi} - \frac{w(X)\,(1-D)\,(Y^\text{obs}-\mu_1(X))}{1-\pi}\right]
\]
\section{Representation of $\delta_w$ in terms of $m(X)$}
\label{appendix:inexp-Fw-2}

Let
\[
\delta_w = E\!\left[\frac{(1-\pi)(1-w(X))\,\tau(X)\,D + \pi\, w(X)\,\tau(X)\,(1-D)}{\pi(1-\pi)}\right]
\]

We can rewrite this quantity as
\[
\delta_w = \underbrace{E\!\left[\frac{(1-w(X))(Y(1)-Y(0))\,D}{\pi}\right]}_{A} + \underbrace{E\!\left[\frac{w(X)(Y(1)-Y(0))(1-D)}{1-\pi}\right]}_{B}.
\]

Let us show that
\[
A = E\!\left[\frac{(1-w(X))(Y(1)-Y(0))\,D}{\pi}\right] = E\!\left[\frac{(1-w(X))}{\pi}\,D\,\frac{Y^\text{obs}-m(X)}{1-e(X)}\right],
\]
that is, factoring out $1/\pi$,
\[
E\big[(1-w(X))(Y(1)-Y(0))\,D\big] = E\!\left[(1-w(X))\,\frac{Y^\text{obs}-m(X)}{1-e(X)}\,D\right]
\]

To do so, let us show that both quantities are equal to $E\big[(1-w(X))\,e(X)\,(E[Y(1)\mid D=1,X] - E[Y(0)\mid D=0,X])\big]$.

On the left-hand side of the equality:
\begin{align*}
E\big[(1-w(X))(Y(1)-Y(0))\,D\big] &= E\big[E[(1-w(X))(Y(1)-Y(0))\,D \mid X]\big]\\
&= E\big[(1-w(X))\,e(X)\,E[Y(1)-Y(0)\mid X]\big]\\
&= E\big[(1-w(X))\,e(X)\,(E[Y(1)\mid X] - E[Y(0)\mid X])\big] \quad \text{by CIA.}
\end{align*}

On the right-hand side of the equality:
\[
E\!\left[(1-w(X))\,\frac{Y^\text{obs}-m(X)}{1-e(X)}\,D\right] = E\!\left[\frac{(1-w(X))}{1-e(X)}\,E[D\,Y^\text{obs} - D\,m(X)\mid X]\right]
\]
Since $D\,Y^\text{obs} = D\big(D\,Y(1) + (1-D)\,Y(0)\big) = D\,Y(1)$, we have
\[
E[D\,Y^\text{obs} - D\,m(X)\mid X] = E[D\,Y(1) - D\,m(X)\mid X] = e(X)\,\big(E[Y(1)\mid X] - E[m(X)\mid X]\big) \quad \text{by CIA.}
\]
Moreover (Assumption~4 and CIA),
\begin{align*}
E[m(X)\mid X] = & E\big[E[Y(1)\,e(X) + Y(0)\,(1-e(X)) \mid X] \,\big|\, X\big] \\ = & e(X)\,E[Y(1)\mid X] + (1-e(X))\,E[Y(0)\mid X],
\end{align*}
so
\begin{align*}
E[Y(1)\mid X] - E[m(X)\mid X] &= E[Y(1)\mid X] - e(X)\,E[Y(1)\mid X] - (1-e(X))\,E[Y(0)\mid X]\\
&= (1-e(X))\,(E[Y(1)\mid X] - E[Y(0)\mid X]).
\end{align*}
Thus
\begin{align*}
E\!\left[(1-w(X))\,\frac{Y^\text{obs}-m(X)}{1-e(X)}\,D\right] &= E\!\left[\frac{1-w(X)}{1-e(X)}\,e(X)\,(1-e(X))\,(E[Y(1)\mid X] - E[Y(0)\mid X])\right]\\
&= E\big[(1-w(X))\,e(X)\,(E[Y(1)\mid X] - E[Y(0)\mid X])\big]
\end{align*}

Both sides are therefore equal, which establishes the expression for $A$. The term $B$ is obtained symmetrically (in $(1-D)$, under $Y(1)\indep D\mid X$), which finally gives
\[
\delta_w = E\!\left[\frac{(1-w(X))}{\pi(1-e(X))}\,D\,(Y^\text{obs}-m(X)) - \frac{w(X)}{(1-\pi)\,e(X)}\,(1-D)\,(Y^\text{obs}-m(X))\right]
\]

\section{Proof for the IPW estimator of the unexplained part}
\label{appendix:ipw}

We start from the counterfactual part of~\eqref{dec:unexpl:as_sum_obs_and_count}. For the term in $Y(0)$ (counterfactual of the treated):
\[
E\!\left[\frac{(1-\pi)(1-w(X))\,Y(0)\,D}{\pi(1-\pi)}\right] = E\!\left[\frac{(1-w(X))\,Y(0)\,D}{\pi}\right] = E[(1-w(X))\,Y(0)\mid D=1]
\]

This proof is directly inspired by \citet{kline2011oaxaca}, in a more general setting.
\begin{align*}
E[(1-w(X))\,Y(0)\mid D=1] &= \int E[(1-w(x))\,Y(0)\mid D=1,X=x]\,dF_{X\mid D=1}(x)\\
&= \int (1-w(x))\,E[Y(0)\mid D=1,X=x]\,dF_{X\mid D=1}(x)\\
&= \int (1-w(x))\,E[Y(0)\mid X=x]\,dF_{X\mid D=1}(x).
\end{align*}

Using Bayes:
\begin{align*}
dF_{X\mid D=1}(x) &= \frac{P(D=1\mid X=x)}{P(D=1)}\,dF_X(x) = \frac{e(x)}{\pi}\,dF_X(x),\\
dF_{X\mid D=0}(x) &= \frac{P(D=0\mid X=x)}{P(D=0)}\,dF_X(x) = \frac{1-e(x)}{1-\pi}\,dF_X(x).
\end{align*}
Hence:
\[
dF_{X\mid D=1}(x) = \frac{e(x)\,(1-\pi)}{(1-e(x))\,\pi}\,dF_{X\mid D=0}(x),
\]
providing that $e(X) < 1$.

Thus:
\begin{align*}
E[(1-w(X))\,Y(0)\mid D=1] &= \int (1-w(x))\,E[Y(0)\mid X=x]\,\frac{e(x)(1-\pi)}{(1-e(x))\pi}\,dF_{X\mid D=0}(x)\\
&= \int (1-w(x))\,E[Y(0)\mid X=x,D=0]\,\frac{e(x)(1-\pi)}{(1-e(x))\pi}\,dF_{X\mid D=0}(x)\\
&= E\!\left[\frac{(1-w(X))\,e(X)}{\pi\,(1-e(X))}\,Y^\text{obs}\,(1-D)\right]
\end{align*}

Symmetrically, we have (providing that $e(X) > 0$):
\[
E\!\left[\frac{w(X)\,Y(1)\,(1-D)}{1-\pi}\right] = E[w(X)\,Y(1)\mid D=0] = E\!\left[\frac{w(X)\,(1-e(X))}{(1-\pi)\,e(X)}\,Y^\text{obs}\,D\right]
\]

Hence, under $0 < e(X) < 1$:
\begin{align*}
\delta_w = {}& E\!\left[\frac{(1-w(X))\,Y^\text{obs}D}{\pi} - \frac{(1-w(X))\,e(X)\,Y^\text{obs}(1-D)}{(1-e(X))\,\pi}\right]\\
& - E\!\left[\frac{w(X)\,Y^\text{obs}(1-D)}{1-\pi} - \frac{w(X)\,(1-e(X))\,Y^\text{obs}D}{e(X)\,(1-\pi)}\right]
\end{align*}

Let us factor: using $D(1-e(X)) - e(X)(1-D) = D - e(X)$,
\[
\delta_w = E\!\left[\left(\frac{1-w(X)}{\pi\,(1-e(X))} + \frac{w(X)}{(1-\pi)\,e(X)}\right) Y^\text{obs}\,(D-e(X))\right]
\]

Application to the case $w(X)=0$ (reference $Y_0$):
\[
\delta_0 = E\!\left[\frac{(D-e(X))\,Y^\text{obs}}{\pi\,(1-e(X))}\right]
\]
Application to the case $w(X)=1$ (reference $Y_1$):
\[
\delta_1 = E\!\left[\frac{(D-e(X))\,Y^\text{obs}}{(1-\pi)\,e(X)}\right]
\]
Application to the case $w(X)=e(X)$ (reference $Y_e$):
\[
\delta_e = E\!\left[\frac{(D-e(X))\,Y^\text{obs}}{\pi\,(1-\pi)}\right]
\]
Application to the case $w(X)=\pi$ (reference $Y_{\pi}$):
\[
\delta_{\pi} = E\!\left[\frac{(1-\pi)\,(D-e(X))\,Y^\text{obs}}{\pi\,(1-e(X))} + \frac{\pi\,(D-e(X))\,Y^\text{obs}}{(1-\pi)\,e(X)}\right]
\]
Application to the case $w(X)=\tfrac{1}{2}$ (reference $Y_{1/2}$):
\[
\delta_{1/2} = E\!\left[\frac{(D-e(X))\,Y^\text{obs}}{2\pi\,(1-e(X))} + \frac{(D-e(X))\,Y^\text{obs}}{2(1-\pi)\,e(X)}\right]
\]

\section{Proof for the AIPW estimator of the unexplained part}

\subsection{Based on $\mu_1(X)$ and $\mu_0(X)$}
\label{appendix:aipw_mu0_mu1}

To construct the AIPW estimator, we begin by expressing the bias of the regression-based estimator. Suppose that we plug in functions $\mu_0(X)$ and $\mu_1(X)$ instead of the true functions $E[Y(1)\mid X,D=1]=\mu_1^*(X)$ and $E[Y(0)\mid X,D=0]=\mu_0^*(X)$. Starting from expression~\eqref{unexp:regression_g0_g1}:
\begin{align*}
\delta_w &= E\!\left[\frac{1-w(X)}{\pi}\,D\,(Y^\text{obs}-\mu_0^*(X))\right] - E\!\left[\frac{w(X)}{1-\pi}\,(1-D)\,(Y^\text{obs}-\mu_1^*(X))\right]\\
&= E\!\left[\frac{1-w(X)}{\pi}\,D\,\big(Y^\text{obs}-\mu_0(X)+\mu_0(X)-\mu_0^*(X)\big)\right]\\
&\quad - E\!\left[\frac{w(X)}{1-\pi}\,(1-D)\,\big(Y^\text{obs}-\mu_1(X)+\mu_1(X)-\mu_1^*(X)\big)\right]\\
&= E\!\left[\frac{1-w(X)}{\pi}\,D\,(Y^\text{obs}-\mu_0(X))\right] + E\!\left[\frac{1-w(X)}{\pi}\,D\,(\mu_0(X)-\mu_0^*(X))\right]\\
&\quad - E\!\left[\frac{w(X)}{1-\pi}\,(1-D)\,(Y^\text{obs}-\mu_1(X))\right] - E\!\left[\frac{w(X)}{1-\pi}\,(1-D)\,(\mu_1(X)-\mu_1^*(X))\right]\\
&= \delta_w + \mathrm{bias}_0 - \mathrm{bias}_1
\end{align*}
where
\[
\mathrm{bias}_0 = E\!\left[\frac{1-w(X)}{\pi}\,D\,(\mu_0(X)-\mu_0^*(X))\right], \qquad
\mathrm{bias}_1 = E\!\left[\frac{w(X)}{1-\pi}\,(1-D)\,(\mu_1(X)-\mu_1^*(X))\right]
\]

$\mathrm{bias}_0$ depends on the unknown $\mu_0^*(X)$. We seek to identify it via the inverse probability weighting strategy in order to correct for it should it arise (i.e., if $\mu_0(X)\neq \mu_0^*(X)$ in expectation, that is, if there is bias).
\begin{align*}
\mathrm{bias}_0 &= E\big[(1-w(X))\,(\mu_0(X)-\mu_0^*(X))\mid D=1\big]\\
&= \int (1-w(x))\,(\mu_0(x)-\mu_0^*(x))\,dF_{X\mid D=1}(x)\\
&= \int (1-w(x))\,(\mu_0(x)-\mu_0^*(x))\,\frac{e(x)(1-\pi)}{(1-e(x))\pi}\,dF_{X\mid D=0}(x) \qquad \text{assuming } e(X)<1\\
&= \int (1-w(x))\,(\mu_0(x)-Y(0))\,\frac{e(x)(1-\pi)}{(1-e(x))\pi}\,dF_{X\mid D=0}(x)\\
&= E\!\left[(1-w(X))\,(\mu_0(X)-Y^\text{obs})\,\frac{e(X)(1-\pi)}{(1-e(X))\pi}\ \middle|\ D=0\right]\\
&= E\!\left[\frac{(1-w(X))\,e(X)}{(1-e(X))\,\pi}\,(\mu_0(X)-Y^\text{obs})\,(1-D)\right]
\end{align*}

Symmetrically, assuming $e(X)>0$:
\[
\mathrm{bias}_1 = E\!\left[\frac{w(X)\,(1-e(X))}{e(X)\,(1-\pi)}\,(\mu_1(X)-Y^\text{obs})\,D\right]
\]

Hence
\begin{align*}
\delta_w &= E\!\left[\frac{1-w(X)}{\pi}\,(Y^\text{obs}-\mu_0(X))\left(D-\frac{e(X)(1-D)}{1-e(X)}\right)\right]\\
&\quad - E\!\left[\frac{w(X)}{1-\pi}\,(Y^\text{obs}-\mu_1(X))\left(1-D-\frac{(1-e(X))D}{e(X)}\right)\right]
\end{align*}

We then obtain
\begin{align*}
\delta_w = &E\!\left[\frac{1-w(X)}{\pi}\,(Y^\text{obs}-\mu_0(X))\,\frac{D-e(X)}{1-e(X)}\right]
 - E\!\left[\frac{w(X)}{1-\pi}\,(Y^\text{obs}-\mu_1(X))\,\frac{e(X)-D}{e(X)}\right]\\
= & E\!\left[\frac{1}{\pi}\,(Y^\text{obs}-\mu_0(X))\,(D-e(X))\,\frac{1-w(X)}{1-e(X)}\right] \\
& + E\!\left[\frac{1}{1-\pi}\,(Y^\text{obs}-\mu_1(X))\,(D-e(X))\,\frac{w(X)}{e(X)}\right]
\end{align*}

Application to the case $w(X)=0$ (reference $Y_0$):
\[
\delta_0 = E\!\left[\frac{1}{\pi}\,(Y^\text{obs}-\mu_0(X))\,\frac{D-e(X)}{1-e(X)}\right]
\]

Application to the case $w(X)=1$ (reference $Y_1$):
\[
\delta_1 = E\!\left[\frac{1}{1-\pi}\,(Y^\text{obs}-\mu_1(X))\,\frac{D-e(X)}{e(X)}\right]
\]

Application to the case $w(X)=e(X)$ (reference $Y_e$):
\[
\delta_e = E\!\left[\frac{1}{\pi}\,(Y^\text{obs}-\mu_0(X))\,(D-e(X))\right] + E\!\left[\frac{1}{1-\pi}\,(Y^\text{obs}-\mu_1(X))\,(D-e(X))\right]
\]

Application to the case $w(X)=\pi$ (reference $Y_{\pi}$):
\[
\delta_{\pi} = E\!\left[\frac{1-\pi}{\pi}\,(Y^\text{obs}-\mu_0(X))\,\frac{D-e(X)}{1-e(X)}\right] + E\!\left[\frac{\pi}{1-\pi}\,(Y^\text{obs}-\mu_1(X))\,\frac{D-e(X)}{e(X)}\right]
\]

Application to the case $w(X)=\tfrac{1}{2}$ (reference $Y_{1/2}$):
\[
\delta_{1/2} = E\!\left[\frac{1}{2\pi}\,(Y^\text{obs}-\mu_0(X))\,\frac{D-e(X)}{1-e(X)}\right] + E\!\left[\frac{1}{2(1-\pi)}\,(Y^\text{obs}-\mu_1(X))\,\frac{D-e(X)}{e(X)}\right]
\]

\subsection{Based on $m(X)$}
\label{appendix:aipw-g2}

Let $m^*(X)$ denote the correct representation of $E[Y^\text{obs}\mid X]$ and $m(X)$ a biased estimate of it that we plug in. Equation~\eqref{unexp:regression_g2} allows us to state, under the assumption $m^*(X)=E[Y^\text{obs}\mid X]$ and under $0<e(X)<1$:
\begin{align*}
\delta_w = {}& E\!\left[\frac{1-w(X)}{\pi(1-e(X))}\,D\,\big(Y^\text{obs}-m^*(X)+m(X)-m^*(X)\big)\right]\\
&- E\!\left[\frac{w(X)}{(1-\pi)\,e(X)}\,(1-D)\,\big(Y^\text{obs}-m(X)+m(X)-m^*(X)\big)\right]\\
= {}& \left[\frac{1-w(X)}{\pi(1-e(X))}\,D\,(Y^\text{obs}-m(X)) - \frac{w(X)}{(1-\pi)\,e(X)}\,(1-D)\,(Y^\text{obs}-m(X))\right]\\
&+ E\!\left[\frac{1-w(X)}{\pi(1-e(X))}\,D\,(m(X)-m^*(X))\right] - E\!\left[\frac{w(X)}{(1-\pi)\,e(X)}\,(1-D)\,(m(X)-m^*(X))\right]\\
 = {}& E\!\left[\frac{1-w(X)}{\pi(1-e(X))}\,D\,(Y^\text{obs}-m(X)) - \frac{w(X)}{(1-\pi)\,e(X)}\,(1-D)\,(Y^\text{obs}-m(X))\right]\\
&+ \underbrace{E\!\left[\frac{1-w(X)}{\pi(1-e(X))}\,D\,(m(X)-m^*(X))\right]}_{A} - \underbrace{E\!\left[\frac{w(X)}{(1-\pi)\,e(X)}\,(1-D)\,(m(X)-m^*(X))\right]}_{B}.
\end{align*}

Let us focus on residual $A$:
\begin{align*}
A &= E\!\left[\frac{1-w(X)}{\pi(1-e(X))}\,D\,(m(X)-m^*(X))\right] = E\!\left[\frac{1-w(X)}{1-e(X)}\,(m(X)-m^*(X))\ \middle|\ D=1\right]\\
&= \int \frac{1-w(x)}{1-e(x)}\,(m(x)-m^*(x))\,dF_{X\mid D=1}(x)\\
&= \int \frac{1-w(x)}{1-e(x)}\cdot\frac{e(x)}{\pi}\,(m(x)-m^*(x))\,dF_X(x) \qquad \text{(Bayes)}\\
&= \int \frac{1-w(x)}{1-e(x)}\cdot\frac{e(x)}{\pi}\,(m(x)-Y^\text{obs})\,dF_X(x) \qquad \text{(law of iterated expectations)}\\
&= E\!\left[\frac{1-w(X)}{1-e(X)}\cdot\frac{e(X)}{\pi}\,(m(X)-Y^\text{obs})\right]
\end{align*}

Symmetrically:
\[
B = E\!\left[\frac{w(X)}{e(X)}\cdot\frac{1-e(X)}{1-\pi}\,(m(X)-Y^\text{obs})\right]
\]

Hence:
\begin{align*}
\delta_w = {}& E\!\left[\frac{1-w(X)}{(1-e(X))\,\pi}\,(Y^\text{obs}-m(X))\,(D-e(X))\right]\\
&+ E\!\left[\frac{w(X)}{e(X)\,(1-\pi)}\,(Y^\text{obs}-m(X))\,\big(1-e(X)-(1-D)\big)\right]\\
= {}& E\!\left[\left(\frac{1-w(X)}{(1-e(X))\,\pi}+\frac{w(X)}{e(X)\,(1-\pi)}\right)(D-e(X))\,(Y^\text{obs}-m(X))\right]
\end{align*}

Application to the case $w(X)=0$ (reference $Y_0$):
\[
\delta_0 = E\!\left[\frac{1}{(1-e(X))\,\pi}\,(Y^\text{obs}-m(X))\,(D-e(X))\right]
\]

Application to the case $w(X)=1$ (reference $Y_1$):
\[
\delta_1 = E\!\left[\frac{1}{e(X)\,(1-\pi)}\,(Y^\text{obs}-m(X))\,(D-e(X))\right]
\]

Application to the case $w(X)=e(X)$ (reference $Y_e$):
\[
\delta_e = E\!\left[\left(\frac{1}{\pi}+\frac{1}{1-\pi}\right)(Y^\text{obs}-m(X))\,(D-e(X))\right]
\]

Application to the case $w(X)=\pi$ (reference $Y_{\pi}$):
\[
\delta_{\pi} = E\!\left[\left(\frac{1-\pi}{(1-e(X))\,\pi}+\frac{\pi}{e(X)\,(1-\pi)}\right)(Y^\text{obs}-m(X))\,(D-e(X))\right]
\]

Application to the case $w(X)=\tfrac{1}{2}$ (reference $Y_{1/2}$):
\[
\delta_{1/2} = E\!\left[\left(\frac{1}{2(1-e(X))\,\pi}+\frac{1}{2\,e(X)\,(1-\pi)}\right)(Y^\text{obs}-m(X))\,(D-e(X))\right]
\]

\section{Neyman-orthogonality for the unexplained part}

\subsection{Score function based on $\mu_1(X)$ and $\mu_0(X)$}
\label{proof:score_g0g1}

Note that \citet{chernozhukov2018double} indicates that Neyman-orthogonality is verified in the case of the Average Treatment Effect on the Treated, corresponding here to our special case $w(X)=0$. The following proof is thus based on similar arguments but generalized to any score for a reference outcome in the family $\mathcal{F}$.
	
	Let us define the score function
	\begin{align*}
		\psi_1(W;\theta,\eta)
		&= \frac{1}{\pi}\,\big(\Yobs-\mu_0(X)\big)\,
		\big(D-e(X)\big)\,\frac{1-w(X)}{1-e(X)} \\[4pt]
		&\quad+\ \frac{1}{1-\pi}\,\big(\Yobs-\mu_1(X)\big)\,
		\big(D-e(X)\big)\,\frac{w(X)}{e(X)} -\ \theta
	\end{align*}
	
	where $W=(\Yobs,D,X)$ is the set of observations, $\theta$ is a low-dimensional parameter, and
	$\eta=((\mu_0(X),\mu_1(X)),\,e(X))$ is a set of high-dimensional nuisance functions, with true
	values $\theta^*=\delta_w$ and $\eta^*=((\mu_0^*(X),\mu_1^*(X)),\,e^*(X))$, with $\mathbb{E}[Y^\text{obs}\mid D=0, X]=\mu_0^*(X)$, $\mathbb{E}[Y^\text{obs}\mid D=1, X]=\mu_1^*(X)$ and $\mathbb{E}[D\mid X]=e^*(X)$. 
    
\paragraph{Case 1: $w(X)$ is a known function, not estimated (e.g., $w(X)$ is a known constant)} \ \\

We have :
\begin{align*}
\psi_1(W;\theta,\eta^*+c(\eta-\eta^*))
&= \frac{1-w(X)}{\pi}\;
   \big(\Yobs-\mu_0^*(X)-c\big(\mu_0(X)-\mu_0^*(X)\big)\big) \\
&\qquad\times\, \big(D-e^*(X)-c\big(e(X)-e^*(X)\big)\big) \\
&\qquad\times\, \frac{1}{1-e^*(X)-c\big(e(X)-e^*(X)\big)} \\[8pt]
&\quad+\ \frac{w(X)}{1-\pi}\;
   \big(\Yobs-\mu_1^*(X)-c\big(\mu_1(X)-\mu_1^*(X)\big)\big) \\
&\qquad\times\, \big(D-e^*(X)-c\big(e(X)-e^*(X)\big)\big) \\
&\qquad\times\, \frac{1}{e^*(X)+c\big(e(X)-e^*(X)\big)} \;-\;\theta.
\end{align*}	

	where $c(\eta-\eta^*)$ measures small deviations to the true nuisance functions, as those obtained by replacing the true values by ML estimates.

To simplify the notations, let us write the last quantity with
	\begin{align*}
		u_0(c) &= \Yobs-\mu_0^*(X)-c\big(\mu_0(X)-\mu_0^*(X)\big)\\
		u_1(c) &= \Yobs-\mu_1^*(X)-c\big(\mu_1(X)-\mu_1^*(X)\big)\\
		v(c)   &= D-e^*(X)-c\big(e(X)-e^*(X)\big)\\
		h(c)   &= 1-e^*(X)-c\big(e(X)-e^*(X)\big)
	\end{align*}
	so that:
	\[
	\Eb\big[\psi_1(W;\theta,\eta^*+c(\eta-\eta^*))\big]
	=\Eb\!\left[\frac{1-w(X)}{\pi}\,\frac{u_0(c)\,v(c)}{h(c)}\right]
	+\Eb\!\left[\frac{w(X)}{1-\pi}\,\frac{u_1(c)\,v(c)}{1-h(c)}\right]-\theta
	\]
	
	We can now derive with respect to $c$.
	\begin{align*}
		\partial_c\,\Eb\big[\psi_1(W;\theta,\eta^*+c(\eta-\eta^*))\big]
		&= \Eb\!\left[\frac{1-w(X)}{\pi}\cdot
		\frac{\big[u_0'(c)\,v(c)+u_0(c)\,v'(c)\big]\,h(c)-u_0(c)\,v(c)\,h'(c)}{h(c)^2}\right]\\[6pt]
		&\ +\ \Eb\!\left[\frac{w(X)}{1-\pi}\cdot
		\frac{\big[u_1'(c)\,v(c)+u_1(c)\,v'(c)\big]\,(1-h(c))+u_1(c)\,v(c)\,h'(c)}{(1-h(c))^2}\right]
	\end{align*}

Where:
\begin{align*}
    u_0'(c) &= -\big(\mu_0(X)-\mu_0^*(X)\big)\\
	u_1'(c) &= -\big(\mu_1(X)-\mu_1^*(X)\big)\\
	v'(c)   &= -\big(e(X)-e^*(X)\big)\\
	h'(c)   &= -\big(e(X)-e^*(X)\big)
\end{align*}

Hence:
\begin{align*}
    u_0'(c) &= \mu_0^*(X)-\mu_0(X)\\
	u_1'(c) &= \mu_1^*(X)-\mu_1(X)\\
	v'(c)   &= e^*(X)-e(X)\\
	h'(c)   &= e^*(X)-e(X)
\end{align*}

that do not depend on $c$.

Evaluated at $c=0$, the derivative is: 
	\begin{align*}
		\partial_{c}\,\Eb\big[\psi_1(W;\theta,\eta^*+c(\eta-\eta^*))\big]\Big|_{c=0}
		&= \Eb\!\left[\frac{1-w(X)}{\pi}\cdot
		\frac{\big[u_0'(0)\,v(0)+u_0(0)\,v'(0)\big]\,h(0)-u_0(0)\,v(0)\,h'(0)}{h(0)^2}\right]\\[6pt]
		&\ +\ \Eb\!\left[\frac{w(X)}{1-\pi}\cdot
		\frac{\big[u_1'(0)\,v(0)+u_1(0)\,v'(0)\big]\,(1-h(0))+u_1(0)\,v(0)\,h'(0)}{(1-h(0))^2}\right]
	\end{align*}

with:
	\begin{align*}
		u_0(0) &= \Yobs-\mu_0^*(X)\\
		u_1(0) &= \Yobs-\mu_1^*(X)\\
		v(0)   &= D-e^*(X)\\
		h(0)   &= 1-e^*(X)
	\end{align*}

\paragraph{First expectation.}
Let us consider the first expectation. Since $v'(0)=h'(0)$,
\begin{align*}
E_1 
&= \Eb\!\left[\frac{1-w(X)}{\pi}\,
   \frac{u_0'(0)\,v(0)\,h(0) + u_0(0)\,v'(0)\,h(0) - u_0(0)\,v(0)\,h'(0)}{h(0)^2}\right] \\[4pt]
&= \Eb\!\left[\frac{1-w(X)}{\pi}\,
   \frac{u_0'(0)\,v(0)\,h(0) + u_0(0)\,v'(0)\,\big(h(0)-v(0)\big)}{h(0)^2}\right]
\end{align*}
Moreover,
\[
h(0)-v(0) \;=\; 1-e^*(X) - \big(D-e^*(X)\big) \;=\; 1-D 
\]
Hence
\[
E_1 \;=\; 
\underbrace{\Eb\!\left[\frac{1-w(X)}{\pi}\,\frac{u_0'(0)\,v(0)}{h(0)}\right]}_{a_1}
\;+\;
\underbrace{\Eb\!\left[\frac{1-w(X)}{\pi}\,\frac{u_0(0)\,v'(0)\,(1-D)}{h(0)^2}\right]}_{b_1}
\]

For the first term,
\begin{align*}
a_1 
&= \Eb\!\left[\frac{1-w(X)}{\pi}\,
   \frac{\big(\mu_0^*(X)-\mu_0(X)\big)\big(D-e^*(X)\big)}{1-e^*(X)}\right] \\[4pt]
&= \Eb\!\left[\,\Eb\!\left[\frac{1-w(X)}{\pi}\,
   \frac{\big(\mu_0^*(X)-\mu_0(X)\big)\big(D-e^*(X)\big)}{1-e^*(X)}
   \,\middle|\, X\right]\right]
   && \text{by the Law of Iterated Expectations} \\[4pt]
&= \Eb\!\left[\frac{1-w(X)}{\pi}\,
   \frac{\mu_0^*(X)-\mu_0(X)}{1-e^*(X)}\;
   \underbrace{\Eb\big[D-e^*(X)\mid X\big]}_{=\,0}\right] \;=\; 0 
\end{align*}

For the second term,
\begin{align*}
b_1 
&= \Eb\!\left[\frac{1-w(X)}{\pi}\,
   \big(\Yobs-\mu_0^*(X)\big)(1-D)\,
   \frac{e^*(X)-e(X)}{\big(1-e^*(X)\big)^2}\right] \\[4pt]
&= \Eb\!\left[\,\Eb\!\left[\frac{1-w(X)}{\pi}\,
   \big(\Yobs-\mu_0^*(X)\big)(1-D)\,
   \frac{e^*(X)-e(X)}{\big(1-e^*(X)\big)^2}
   \,\middle|\, X\right]\right] \\[4pt]
&= \Eb\!\left[\frac{1-w(X)}{\pi}\,
   \frac{e^*(X)-e(X)}{\big(1-e^*(X)\big)^2}\;
   \Eb\!\left[\big(\Yobs-\mu_0^*(X)\big)(1-D)\,\middle|\, X\right]\right],
\end{align*}
with
\begin{align*}
\Eb\!\left[\big(\Yobs-\mu_0^*(X)\big)(1-D)\,\middle|\, X\right]
&= \Eb\!\left[\big(Y(0)-\mu_0^*(X)\big)(1-D)\,\middle|\, X\right]
   \quad \text{by definition of } \Yobs \\[2pt]
&= \Eb\!\left[Y(0)-\mu_0^*(X)\,\middle|\, X,\,D=0\right]\,P(D=0\mid X)
   \\& \text{by the Law of Iterated Expectations} \\[2pt]
&= 0
\end{align*}
Hence $b_1=0$, which implies $E_1=0$ since $a_1=0$.

\paragraph{Second expectation.}
Let us consider the second expectation and exploit similar arguments. Since $v'(0)=h'(0)$,
\begin{align*}
E_2 
&= \Eb\!\left[\frac{w(X)}{1-\pi}\,
   \frac{u_1'(0)\,v(0)\,\big(1-h(0)\big)}{\big(1-h(0)\big)^2}\right]
 + \Eb\!\left[\frac{w(X)}{1-\pi}\,
   \frac{u_1(0)\,v'(0)\,\big(1-h(0)+v(0)\big)}{\big(1-h(0)\big)^2}\right]
\end{align*}
Moreover,
\[
1-h(0)+v(0) \;=\; 1-\big(1-e^*(X)\big) + D - e^*(X) \;=\; D,
\qquad\text{and}\qquad
1-h(0) \;=\; e^*(X)
\]
Hence
\[
E_2 \;=\;
\underbrace{\Eb\!\left[\frac{w(X)}{1-\pi}\,\frac{u_1'(0)\,v(0)}{e^*(X)}\right]}_{a_2}
\;+\;
\underbrace{\Eb\!\left[\frac{w(X)}{1-\pi}\,\frac{u_1(0)\,v'(0)\,D}{e^*(X)^2}\right]}_{b_2}
\]
Similarly, by the Law of Iterated Expectations,
\[
a_2 
= \Eb\!\left[\frac{w(X)}{1-\pi}\,
  \frac{\mu_1^*(X)-\mu_1(X)}{e^*(X)}\;
  \underbrace{\Eb\big[D-e^*(X)\mid X\big]}_{=\,0}\right] \;=\; 0 
\]
Moreover,
\[
b_2 
= \Eb\!\left[\frac{w(X)\,\big(e^*(X)-e(X)\big)}{(1-\pi)\,e^*(X)^2}\;
  \Eb\!\left[\big(\Yobs-\mu_1^*(X)\big)D\,\middle|\, X\right]\right],
\]
with
\begin{align*}
\Eb\!\left[\big(\Yobs-\mu_1^*(X)\big)D\,\middle|\, X\right]
&= \Eb\!\left[\big(Y(1)-\mu_1^*(X)\big)D\,\middle|\, X\right]
   \quad \text{by definition of } \Yobs \\[2pt]
&= \Eb\!\left[Y(1)-\mu_1^*(X)\,\middle|\, X,\,D=1\right]\,P(D=1\mid X)
   \\ & \text{by the Law of Iterated Expectations} \\[2pt]
&= 0
\end{align*}
Hence $E_2 = a_2 + b_2 = 0+0 = 0$.

\vspace{0.5cm}
Combining the two expectations,
\[
\partial_c\,\Eb\big[\psi_1(W;\theta^*,\eta^*+c(\eta-\eta^*))\big]\Big|_{c=0}
\;=\; E_1 + E_2 \;=\; 0
\]
Hence, the score $\psi_1$ is Neyman-orthogonal.

Finally, it is important to note the underlying mechanisms. Errors in the estimation of the outcome models are offset by the correct specification of the propensity score model (which cancels out $a_1$ and $a_2$). Conversely, errors in the estimation of the propensity score model, $e(X)-e^*(X)$, are offset by the correct specification of the outcome models (which cancels out $b_1$ and $b_2$). In the literature on treatment effects, there is generally a single outcome model to estimate because the focus is on the treatment effect on the treated or untreated groups, corresponding to our cases $w(X)=1$ or $w(X)=0$ here. However, for other reference outcomes, associated with other functions $w(\cdot)$, it will be necessary to correctly specify both entire functions $\mu_0(X)$ and $\mu_1(X)$. In other words, it will be necessary to simultaneously and correctly estimate the outcome models in both subgroups. If a practitioner wants to reduce this to the estimation of a single function, they can regress the observed outcome simultaneously on $X$ and $D$, in a model allowing for interactions between the group dummy and the covariates.

\paragraph{Case 2: $w(X)=e^*(X)$} \ \\

This case is treated separately, as the weight $w(X)$ is assumed equal to the unknown function $e^*(X)$. It would therefore be inconsistent to treat $w(X)$ as fixed along the perturbation path: since it coincides with a nuisance function, it must itself be replaced by $e^*(X)+c\big(e(X)-e^*(X)\big)$, and differentiated with respect to $c$ accordingly. In this case, however, the score function simplifies considerably. Substituting $w_c(X) = e^*(X)+c\big(e(X)-e^*(X)\big) = 1-h(c)$ into $\psi_1$, so that $1-w_c(X)=h(c)$, the two weights $\frac{1-w_c(X)}{h(c)}$ and $\frac{w_c(X)}{1-h(c)}$ appearing in $\psi_1$ both collapse to $1$, and the score function simplifies to:

	\begin{align*}
		\psi_1(W;\theta^*,\eta)
		&= \frac{1}{\pi}\,\big(\Yobs-\mu_0(X)\big)\,
		\big(D-e(X)\big) \\[4pt]
		&\quad+\ \frac{1}{1-\pi}\,\big(\Yobs-\mu_1(X)\big)\,
		\big(D-e(X)\big) -\ \theta^*
	\end{align*}

	Along the perturbation path $\eta^*+c(\eta-\eta^*)$, and keeping the notations
	$u_0(c)$, $u_1(c)$ and $v(c)$ introduced above, this gives:
	\begin{align*}
		\psi_1\big(W;\theta^*,\eta^*+c(\eta-\eta^*)\big)
		&= \underbrace{\frac{u_0(c)\,v(c)}{\pi}}_{A(c)}
		 + \underbrace{\frac{u_1(c)\,v(c)}{1-\pi}}_{B(c)} - \theta^*
	\end{align*}
    
\begin{align*}
A'(c) &= \frac{1}{\pi}\Big(u_0'(c)\,v(c) + u_0(c)\,v'(c)\Big) \\
A'(0) &= \frac{1}{\pi}\Big(u_0'(0)\,v(0) + u_0(0)\,v'(0)\Big) \\
A'(0) &= \frac{1}{\pi}\Big(\big(\mu_0^*(X)-\mu_0(X)\big)\big(D-e^*(X)\big) 
         + \big(\Yobs-\mu_0^*(X)\big)\big(e^*(X)-e(X)\big)\Big)
\end{align*}
\begin{align*}
B'(c) &= \frac{1}{1-\pi}\Big(u_1'(c)\,v(c) + u_1(c)\,v'(c)\Big) \\
B'(0) &= \frac{1}{1-\pi}\Big(u_1'(0)\,v(0) + u_1(0)\,v'(0)\Big) \\
B'(0) &= \frac{1}{1-\pi}\Big(\big(\mu_1^*(X)-\mu_1(X)\big)\big(D-e^*(X)\big) 
         + \big(\Yobs-\mu_1^*(X)\big)\big(e^*(X)-e(X)\big)\Big)
\end{align*}

Using the Law of Iterated Expectations,
\begin{align*}
\Eb[A'(0)] 
&= \Eb\big[\Eb[A'(0)\mid X]\big] \\
&= \Eb\!\Bigg[\frac{1}{\pi}\Big(\big(\mu_0^*(X)-\mu_0(X)\big)\,\underbrace{\Eb[D-e^*(X)\mid X]}_{=\,0} 
   \\ 
   &+ \Eb\big[(\Yobs-\mu_0^*(X))\,(e^*(X)-e(X))\mid X\big]\Big)\Bigg] \\
&= \Eb\!\left[\frac{e^*(X)-e(X)}{\pi}\;\Eb\big[\Yobs-\mu_0^*(X)\mid X\big]\right],
\end{align*}
and similarly,
\begin{align*}
\Eb[B'(0)] 
=& \Eb\big[\Eb[B'(0)\mid X]\big] \\
=& \Eb\!\Bigg[\frac{1}{1-\pi}\Big(\big(\mu_1^*(X)-\mu_1(X)\big)\,\underbrace{\Eb[D-e^*(X)\mid X]}_{=\,0} 
   \\ &+ \Eb\big[(\Yobs-\mu_1^*(X))\,(e^*(X)-e(X))\mid X\big]\Big)\Bigg] \\
=& \Eb\!\left[\frac{e^*(X)-e(X)}{1-\pi}\;\Eb\big[\Yobs-\mu_1^*(X)\mid X\big]\right]
\end{align*}

\[
\Eb\big[\Yobs-\mu_0^*(X)\mid X\big] = e^*(X)\,(\mu_1^*(X)-\mu_0^*(X)), 
\]
\[
\Eb\big[\Yobs-\mu_1^*(X)\mid X\big] = -\big(1-e^*(X)\big)\,(\mu_1^*(X)-\mu_0^*(X)),
\]
we obtain
\[
\Eb[A'(0)] = \Eb\!\left[\frac{\big(e^*(X)-e(X)\big)\,e^*(X)\,(\mu_1^*(X)-\mu_0^*(X))}{\pi}\right],
\]
\[
\Eb[B'(0)] = -\,\Eb\!\left[\frac{\big(e^*(X)-e(X)\big)\,\big(1-e^*(X)\big)\,(\mu_1^*(X)-\mu_0^*(X))}{1-\pi}\right]
\]
These two terms do not cancel in general, since $\pi$ and $1-\pi$ are fixed constants rather than functions of $X$: the score obtained by substituting $w(X)=e(X)$ is therefore not Neyman-orthogonal.

\subsection{Score function based on $m(X)$}
\label{appendix:score-aipw-g2}

Alternatively, we derived an AIPW estimator based on the pooled regression of $\Yobs$ on $X$ (omitting $D$), following the spirit of \cite{neumark1988employers}'s approach. We can derive the corresponding score function:
\[
\psi_2(W;\theta,\eta) = \left(\frac{1-w(X)}{(1-e(X))\,\pi}+\frac{w(X)}{e(X)\,(1-\pi)}\right)\big(D-e(X)\big)\,\big(\Yobs-m(X)\big) - \theta,
\]
where $W=(\Yobs,D,X)$ is the set of observations, $\theta$ is a low-dimensional parameter, and $\eta=(m(X),\,e(X))$ is a set of high-dimensional nuisance functions, with true values $\theta^*=\delta_w$ and $\eta^*=(m^*(X),\,e^*(X))$, with $\mathbb{E}[\Yobs\mid X]=m^*(X)$ and $\mathbb{E}[D\mid X]=e^*(X)$. 

\paragraph{Case 1: $w(X)$ is a known function, not estimated (e.g., $w(X)$ is a known constant)} \ \\

We have:

\begin{align*}
\psi_2(W;\theta,\eta^*+c(\eta-\eta^*))
&= \Bigg(\frac{1-w(X)}{\pi\big(1-e^*(X)-c(e(X)-e^*(X))\big)} \\ & +\frac{w(X)}{(1-\pi)\big(e^*(X)+c(e(X)-e^*(X))\big)}\Bigg) \\
&\quad\times\, \big(D-e^*(X)-c(e(X)-e^*(X))\big) \\
&\quad\times\, \big(\Yobs-m^*(X)-c(m(X)-m^*(X))\big) \;-\;\theta,
\end{align*}
where $c(\eta-\eta^*)$ measures small deviations to the true nuisance functions, as those obtained by replacing the true values by ML estimates.

To simplify the notations, let us write the last quantity with
\begin{align*}
u_2(c) &= \Yobs-m^*(X)-c\big(m(X)-m^*(X)\big)\\
v(c)   &= D-e^*(X)-c\big(e(X)-e^*(X)\big)\\
h(c)   &= 1-e^*(X)-c\big(e(X)-e^*(X)\big)
\end{align*}
so that:
\[
\Eb\big[\psi_2(W;\theta,\eta^*+c(\eta-\eta^*))\big]
=\Eb\!\underbrace{\left[\frac{1-w(X)}{\pi}\,\frac{u_2(c)\,v(c)}{h(c)}\right]}_{A(c)}
+\Eb\!\underbrace{\left[\frac{w(X)}{1-\pi}\,\frac{u_2(c)\,v(c)}{1-h(c)}\right]}_{B(c)}-\ \theta
\]

\[
A'(c) = \frac{1-w(X)}{\pi}\;\frac{\partial}{\partial c}\!\left(\frac{u_2(c)\,v(c)}{h(c)}\right)
= \frac{1-w(X)}{\pi}\;\frac{u_2'(c)\,v(c)\,h(c) + u_2(c)\,v'(c)\,h(c) - u_2(c)\,v(c)\,h'(c)}{h(c)^2},
\]
where
\begin{align*}
u_2'(c) &= m^*(X)-m(X) \\
v'(c) &= e^*(X)-e(X) \\
h'(c) &= e^*(X)-e(X)
\end{align*}

Thus, using $v'(c)=h'(c)$,
\[
A'(c) = \frac{1-w(X)}{\pi}\;\frac{u_2'(c)\,v(c)\,h(c) + u_2(c)\,v'(c)\,\big(h(c)-v(c)\big)}{h(c)^2}
\]
Since $h(c)-v(c) = 1-D$,
\[
\Eb\big[A'(c)\big] 
= \Eb\!\left[\frac{1-w(X)}{\pi}\,\frac{u_2'(c)\,v(c)}{h(c)}\right]
+ \Eb\!\left[\frac{1-w(X)}{\pi}\,\frac{u_2(c)\,v'(c)\,(1-D)}{h(c)^2}\right]
\]

Evaluated at $c=0$, this quantity yields
\[
\Eb\big[A'(0)\big] = 
\underbrace{\Eb\!\left[\frac{1-w(X)}{\pi}\,\frac{u_2'(0)\,v(0)}{h(0)}\right]}_{a_1}
\;+\;
\underbrace{\Eb\!\left[\frac{1-w(X)}{\pi}\,\frac{u_2(0)\,v'(0)\,(1-D)}{h(0)^2}\right]}_{b_1}
\]

For the first term,
\begin{align*}
a_1 
&= \Eb\!\left[\frac{1-w(X)}{\pi}\,
   \frac{\big(m^*(X)-m(X)\big)\big(D-e^*(X)\big)}{1-e^*(X)}\right] \\[4pt]
&= \Eb\!\left[\,\Eb\!\left[\frac{1-w(X)}{\pi}\,
   \frac{\big(m^*(X)-m(X)\big)\big(D-e^*(X)\big)}{1-e^*(X)}
   \,\middle|\, X\right]\right]
   \\ & \text{by the Law of Iterated Expectations} \\[4pt]
&= \Eb\!\left[\frac{1-w(X)}{\pi}\,
   \frac{m^*(X)-m(X)}{1-e^*(X)}\;
   \underbrace{\Eb\big[D-e^*(X)\mid X\big]}_{=\,0}\right] \;=\; 0 
\end{align*}

For the second term,
\begin{align*}
b_1 
&= \Eb\!\left[\frac{1-w(X)}{\pi}\,
   \big(\Yobs-m^*(X)\big)(1-D)\,
   \frac{e^*(X)-e(X)}{\big(1-e^*(X)\big)^2}\right] \\[4pt]
&= \Eb\!\left[\,\Eb\!\left[\frac{1-w(X)}{\pi}\,
   \big(\Yobs-m^*(X)\big)(1-D)\,
   \frac{e^*(X)-e(X)}{\big(1-e^*(X)\big)^2}
   \,\middle|\, X\right]\right]
   \\ & \text{applying again the Law of Iterated Expectations} \\[4pt]
&= \Eb\!\left[\frac{1-w(X)}{\pi}\,
   \frac{e^*(X)-e(X)}{\big(1-e^*(X)\big)^2}\;
   \Eb\!\left[\big(\Yobs-m^*(X)\big)(1-D)\,\middle|\, X\right]\right],
\end{align*}

Importantly, note that:
\[
\Eb\!\left[\big(\Yobs-m^*(X)\big)(1-D)\,\mid \, X\right] = \Eb\big[(Y(0)-m^*(X))(1-D)\,\mid \, X\big]
\]

which, in general, is not equal to zero. Intuitively, this quantity averages a contribution over the whole population: an individual in group $D=1$ contributes nothing, while an individual in group $D=0$ contributes the gap between their potential outcome $Y(0)$ and $m^*(X)=\Eb[Y_e\mid X]$. Since $m^*(X)$ incorporates information from $Y(1)$ as well, it does not coincide with $\Eb[Y(0)\mid X]$ in general. Let us check nonetheless that there is no compensation between $A'(0)$ and $B'(0)$.

$$B(c) = \dfrac{w(X)}{1-\pi}\,\dfrac{u_2(c)\,v(c)}{1-h(c)}$$

Let $k_1(c)=u_2(c)\,v(c)$ and $k_2(c)=1-h(c)$.

\[
\frac{\partial}{\partial c}\!\left(\frac{k_1(c)}{k_2(c)}\right) = \frac{k_1'(c)\,k_2(c) - k_1(c)\,k_2'(c)}{k_2(c)^2},
\]
with
\[
k_1'(c) = u_2'(c)\,v(c) + u_2(c)\,v'(c), \qquad k_2'(c) = -h'(c).
\]
Hence
\[
\frac{\partial}{\partial c}\!\left(\frac{k_1(c)}{k_2(c)}\right) 
= \frac{u_2'(c)\,v(c)\,\big(1-h(c)\big) + u_2(c)\,v'(c)\,\big(1-h(c)\big) + u_2(c)\,v(c)\,h'(c)}{\big(1-h(c)\big)^2}.
\]
Using $v'(c)=h'(c)$,
\[
= \frac{u_2'(c)\,v(c)\,\big(1-h(c)\big) + u_2(c)\,v'(c)\,\big(1-h(c)+v(c)\big)}{\big(1-h(c)\big)^2}.
\]
Moreover,
\[
1-h(c)+v(c) = 1-\Big(1-e^*(X)-c\big(e(X)-e^*(X)\big)\Big) + D-e^*(X)-c\big(e(X)-e^*(X)\big) = D.
\]
Hence
\[
\frac{\partial}{\partial c}\!\left(\frac{k_1(c)}{k_2(c)}\right) = \frac{u_2'(c)\,v(c)\,\big(1-h(c)\big) + u_2(c)\,v'(c)\,D}{\big(1-h(c)\big)^2},
\]
so that
\[
B'(c) = \frac{w(X)}{1-\pi}\;\frac{u_2'(c)\,v(c)\,\big(1-h(c)\big) + u_2(c)\,v'(c)\,D}{\big(1-h(c)\big)^2}
= \frac{w(X)}{1-\pi}\,\frac{u_2'(c)\,v(c)}{1-h(c)} + \frac{w(X)}{1-\pi}\,\frac{u_2(c)\,v'(c)\,D}{\big(1-h(c)\big)^2}.
\]

\[
\Eb\big[B'(0)\big] = \underbrace{\Eb\!\left[\frac{w(X)}{1-\pi}\,\frac{u_2'(0)\,v(0)}{1-h(0)}\right]}_{a_2} + \underbrace{\Eb\!\left[\frac{w(X)}{1-\pi}\,\frac{u_2(0)\,v'(0)\,D}{\big(1-h(0)\big)^2}\right]}_{b_2}
\]

The first term $a_2$ is similar to $a_1$, except that the factor in front of $u_2'(0)\,v(0)$ is slightly different. We use the same arguments to show that $a_2=0$.

The second term $b_2$ is also similar to $b_1$, so it's non-zero in general. For $w(X)=0$, $\partial_c\,\Eb\big[\psi_2(W;\theta^*,\eta^*+c(\eta-\eta^*))\big]\big|_{c=0} = b_1 \neq 0$ and for $w(X)=1$, $\partial_c\,\Eb\big[\psi_2(W;\theta^*,\eta^*+c(\eta-\eta^*))\big]\big|_{c=0} = b_2 \neq 0$. Thus, in general, Neyman-orthogonality is not verified. For other cases $w(X)\notin\{0,1\}$ and $w(X)\neq e(X)$:
\[
\Eb\big[A'(0)\big] + \Eb\big[B'(0)\big] = b_1 + b_2
\]
\[
= \Eb\!\left[\frac{1-w(X)}{\pi}\,\frac{u_2(0)\,v'(0)\,(1-D)}{h(0)^2}\right] + \Eb\!\left[\frac{w(X)}{1-\pi}\,\frac{u_2(0)\,v'(0)\,D}{\big(1-h(0)\big)^2}\right]
\]

\paragraph{Case 2: $w(X)=e^*(X)$} \ \\

This case is treated separately, as the weight $w(X)$ is assumed equal to the unknown function $e^*(X)$. It would therefore be inconsistent to treat $w(X)$ as fixed along the perturbation path: since it coincides with a nuisance function, it must itself be replaced by $e^*(X)+c\big(e(X)-e^*(X)\big)$, and differentiated with respect to $c$ accordingly. In this case, however, the score function simplifies considerably. Substituting $w_c(X) = e^*(X)+c\big(e(X)-e^*(X)\big) = 1-h(c)$ into $\psi_2$, so that $1-w_c(X)=h(c)$, the two weights $\frac{1-w_c(X)}{h(c)}$ and $\frac{w_c(X)}{1-h(c)}$ appearing in $\psi_2$ both collapse to $1$, and the score function simplifies to:

\[
\psi_2(W;\theta,\eta^*+c(\eta-\eta^*)) = \left( \frac{1}{\pi}+\frac{1}{1-\pi} \right) \,u_2(c)\,v(c) - \theta.
\]

Hence:
\[
\partial_c\,\psi_2(W;\theta,\eta^*+c(\eta-\eta^*)) = \left( \frac{1}{\pi}+\frac{1}{1-\pi} \right) \big(u_2'(c)\,v(c) + u_2(c)\,v'(c)\big).
\]
Evaluated at $c=0$,
\begin{align*}
\partial_c\,\Eb\Big[\psi_2(W;\theta,\eta^*+c(\eta-\eta^*))\Big]\big|_{c=0} 
&= \Eb\Bigg[ \left( \frac{1}{\pi}+\frac{1}{1-\pi} \right) \Big(\big(m^*(X)-m(X)\big)\big(D-e^*(X)\big) \\ & + \big(\Yobs-m^*(X)\big)\big(e^*(X)-e(X)\big)\Big)\Bigg]
\end{align*}

Using the law of iterated expectations, we can use $\Eb[D-e^*(X)\mid X]=0$ and $\Eb[Y^\text{obs} - m^*(X) \mid X]=0$ to verify Neyman-orthogonality.

\section{Facts about shrinkage}

\subsection{$\delta_w$ when $\tau(X)$ is constant (absence of heterogeneity)}
\label{proof:no_heterogeneity}

When $E[Y(1)-Y(0)\mid X] = \tau$ (constant), i.e., in the absence of heterogeneity, from equation~\eqref{unexplained_part_as_function_of_tau}:
\[
\delta_w = E\!\left[\frac{(1-\pi)(1-w(X))\,\tau(X)\,D + \pi\,w(X)\,\tau(X)\,(1-D)}{\pi(1-\pi)}\right]
\]
Since $\tau(X)=\tau$ is constant, this simplifies to
\[
\delta_w = \frac{\tau}{\pi(1-\pi)}\,E\big[(1-\pi)(1-w(X))\,D + \pi\,w(X)\,(1-D)\big]
\]

\textbf{Case 1: $w(X)=k$, a constant.}
\begin{align*}
\delta_w &= \frac{\tau}{\pi(1-\pi)}\,E\big[(1-\pi)(1-k)D + \pi k(1-D)\big]\\
&= \frac{\tau}{\pi(1-\pi)}\,\big[(1-\pi)(1-k)\pi + \pi k(1-\pi)\big]\\
&= \frac{\tau}{\pi(1-\pi)}\,\pi(1-\pi)\big[(1-k)+k\big]\\
&= \tau
\end{align*}
This holds for any constant $k$, and in particular for $k=0,1,\pi,\tfrac{1}{2}$. This implies $\delta_0=\delta_1=\delta_{\pi}=\delta_{1/2}=\tau$.

\textbf{Case 2: $w(X)=e(X)$.}
\[
\delta_w = \tau \, \frac{E\big[(1-\pi)(1-e(X))D + \pi\,e(X)(1-D)\big]}{\pi(1-\pi)}
\]

Using the law of iterated expectations:

\begin{align*}
\Eb\big[(1-\pi)(1-e(X))D + \pi\,e(X)(1-D)\big] &= \Eb\big[(1-\pi)(1-e(X))e(X) + \pi\,e(X)(1-e(X))\big] \\
&= \Eb\big[(1-e(X))e(X)(1-\pi +\pi)\big] \\
&= \Eb\big[(1-e(X))e(X)\big]
\end{align*}

This yields the relationship:

\begin{align*}
    \delta_e = \tau \, \frac{\Eb\big[(1-e(X))e(X)\big]}{\pi(1-\pi)}
\end{align*}

The numerator $\Eb\big[(1-e(X))e(X)\big] \geq 0$ as both $1-e(X)$ and $e(X)$ are nonnegative (almost surely). Moreover:

\begin{align*}
    V(e(X)) &= \Eb[e(X)^2] - \Eb[e(X)]^2 \\
    \Eb[e(X)(1-e(X))] &= \Eb[e(X)-e(X)^2] \\
    &=\Eb[e(X)] - \Eb[e(X)^2] \\
    &= \Eb[e(X)] - V(e(X)) - \Eb[e(X)]^2 \\
    &= \pi - \pi^2 - V(e(X)) \\
    &= \pi \, (1-\pi) - V(e(X)) \\
\text{Thus} \quad \frac{\Eb\big[(1-e(X))e(X)\big]}{\pi(1-\pi)} &= 1 - R^2 \\
\text{With} \quad R^2 &= \frac{V(e(X))}{\pi \, (1-\pi)} = \frac{V(\Eb[D\mid X])}{V(D)}
\end{align*}

Which yields:
\begin{align*}
    \delta_e = \tau \, \big(1-R^2\big)
\end{align*}

$R^2$ being a ratio of variances, it is nonnegative. Thus, $1-R^2\leq 1$. We have shown above that $1-R^2\geq 0$. Thus, $1-R^2$ is a factor shrinking $\tau$ towards $0$.

Note that, in general, $\tau \, \big(1-R^2\big)$ is different from $\tau$. For $\delta_e$ to be equal to $\tau$, assuming $\tau\neq 0$, we need:
\begin{align*}
    R^2 &= 0 \\
    V(e(X))&= 0
\end{align*}

Which means that $e(X)$ is constant (almost surely). In other words, $e(X)=\Eb[D\mid X] = E[D] = \pi$ (almost surely). This implies that the composition in both groups is identical. In this case, $Y_e=Y_{\pi}$ and:

\[
\delta_0=\delta_1=\delta_e=\delta_{\pi}=\delta_{1/2}.
\]

\subsection{Expression of $\delta_e$ in terms of $R^2_D$ under ignorability}
\label{proof:delta_2_in_terms_of_R2}

We have shown in Appendix~\ref{proof:no_heterogeneity} that:

$$
\frac{\Eb[(1-e(X))\,e(X)]}{\pi (1-\pi)}=1-R^2_D
$$

Moreover, we have shown in Appendix~\ref{proof:counterexample1} that:

$$
\delta_e = \Eb\Bigg[ \tau(X)\frac{(1-e(X))\,e(X)}{\pi\,(1-\pi)}\Bigg]
$$

Thus,

\begin{align*}
\frac{\delta_e}{1-R^2_D} &= \Eb\Bigg[ \tau(X)\frac{(1-e(X))\,e(X)}{\Eb[(1-e(X))\,e(X)]}\Bigg] \\ 
&= \frac{\Eb[\tau(X) (1-e(X))\,e(X)]}{\Eb[(1-e(X))\,e(X)]}     
\end{align*}

Hence:

\begin{align*}
\delta_e &= (1-R^2_D) \times \frac{\Eb[\tau(X) (1-e(X))\,e(X)]}{\Eb[(1-e(X))\,e(X)]}     
\end{align*}

\subsection{Counterexample in presence of heterogeneity showing $\delta_e \in ]\delta_0, \delta_1[$}
\label{proof:counterexample1}

First, let us rewrite the unexplained shares under the ignorability assumption:

\begin{align}
    \label{formula_tau_x_factor}
    \delta_0 &= \Eb \left[ \tau(X) \frac{(1-\pi) \, e(X)}{\pi \, (1-\pi)}\right]\\
    \delta_1 &= \Eb \left[ \tau(X) \frac{\big(1-e(X)\big) \, \pi}{\pi \, (1-\pi)}\right] \nonumber \\
    \delta_e &= \Eb \left[ \tau(X) \frac{\big(1-e(X)\big) \, e(X)}{\pi \, (1-\pi)}\right] \nonumber 
\end{align}

Assume $\delta_0 = \frac{1}{6}$, $\delta_e = \frac{1}{3}$, $\delta_1=\frac{1}{2}$, which implies in this particular case that $\delta_0 < \delta_e < \delta_1$. 

Using the law of total probability:

\begin{align*}
    \pi &= P(D=1) \\
    &= e(0) \times \frac{1}{4} + e(1) \times \frac{3}{4} \\
    &= \frac{1}{2} \times \frac{1}{4} + \frac{5}{6} \times \frac{3}{4} \\
    &= \frac{6}{8} \\
    &= \frac{3}{4}
\end{align*}

Since $\tau(1)=0$, only observations with $X=0$ contribute to the expectations. Applying the law of total probability again, only the expectations conditioned on $X=0$ remain:

\begin{align*}
    \delta_0 &= \tau(0) \frac{e(0)}{\pi} \times (1-\pi) \\
    &= \frac{1/2}{3/4} \times \frac{1}{4} \\
    &= \frac{1}{6} \\
    \delta_1 &= \tau(0) \frac{1-e(0)}{1-\pi} \times (1-\pi) \\
    &= 1 - e(0) \\
    &= \frac{1}{2} \\
    \delta_e &= \tau(0) \frac{(1-e(0)) e(0)}{\pi (1-\pi)} \times (1-\pi) \\
    &= \frac{\frac{1}{2} \times \frac{1}{2}}{3\times 4} = \frac{1}{3}
\end{align*}

In this particular case: $\delta_0 < \delta_e < \delta_1$.

\subsection{Counterexample in presence of heterogeneity with $|\delta_e| > max(|\delta_0|, |\delta_1|)$}
\label{proof:counterexample2}

Let us assume that $\mathcal{X} = \{0,1,2\}$, with $P(X=0)=0.2$, $P(X=1)=0.2$, $P(X=2)=0.6$. Moreover, let us assume that $e(0)=0.2$, $e(1)=0.8$, $e(2)=0.5$. In addition, $\tau(0)=-1$, $\tau(1)=-1$, $\tau(2)=1$. Let us derive $\pi$ in this setup.

\begin{align*}
\pi &= P(D=1\mid X=0)P(X=0) + P(D=1\mid X=1)P(X=1) + P(D=1\mid X=2)P(X=2) \\
&= 0.2\times 0.2 + 0.8\times 0.2 + 0.5\times 0.6 \\
&= 0.2\times(0.2+0.8) + 0.3 \\
&= 0.2 + 0.3 \\
&= 0.5
\end{align*}

which implies that $\dfrac{1}{\pi(1-\pi)} = 4$

\begin{align*}
\delta_0 &= 4 \times \mathbb{E}[\tau(X)\times 0.5 \times e(X)] = 2\times \mathbb{E}[T(X)e(X)] \\
\delta_1 &= 2 \times \mathbb{E}[\tau(X)(1-e(X))] \\
\delta_e &= 4 \times \mathbb{E}[(\tau(X)(1-e(X)))e(X)]
\end{align*}

\begin{align*}
\mathbb{E}[\tau(X)e(X)] &= \tau(0)e(0)\times 0.2 + \tau(1)e(1)\times 0.2 + \tau(2)e(2)\times 0.6 \\
&= -1\times 0.2\times 0.2 + (-1)\times 0.8\times 0.2 + 1\times 0.5\times 0.6 \\
&= (-1)\times 0.2 + 0.3\times 1 = -0.2+0.3 = 0.1
\end{align*}

Hence $\delta_0 = 2\times 1 = 0.2$.

\begin{align*}
\mathbb{E}[\tau(X)(1-e(X))] =& \ \tau(0)\times(1-e(0))\times 0.2 \\
& + \tau(1)(1-e(1))\times 0.2 + \tau(2)\times(1-e(2))\times 0.6 \\
=& -1\times 0.2\times(0.8+0.2) + 0.3 = -0.2+0.3 = 0.1
\end{align*}

Hence $\delta_1 = 0.2$.

\begin{align*}
\mathbb{E}[\tau(X)(1-e(X))e(X)] = & \ \tau(0)(1-0.2)0.2\times 0.2 \\ &+ \ \tau(1)(1-0.8)\times 0.8\times 0.2 + \tau(2)(1-0.5)\times 0.5\times 0.6 \\
= & -1\times 0.2(0.2-0.04+0.8-0.64) + \dfrac{6}{4}
\end{align*}

$\Rightarrow \delta_e = -0.4\times 2\times 0.32 + 6 = 0.344$

\section{Residuals on residuals regression}
\subsection{Equality between the target value of the residuals on residuals regression coefficient and $\frac{\delta_e}{1-R^2_D}$ under ignorability}
\label{proof:residuals_on_residuals}

Let $\quad \tilde{D}=D-\mathbb{E}[D \mid X]=D-e(X)$ and $\widetilde{Y}^{\text{obs}}=Y^{\text{obs}}-\mathbb{E}[Y^{\text{obs}} \mid X]$
 
The regression coefficient of residuals on residuals is:
 
$$\delta_{\text{res}}=\frac{\mathbb{E}[\tilde{D} \tilde{Y}^{\text{obs}}]}{\mathbb{E}[\tilde{D}^2]}$$
 
Let us show that, under ignorability:
 
$$\delta_{\text{res}}=\frac{\mathbb{E}[\tau(X) e(X)(1-e(X))]}{\mathbb{E}[e(X)(1-e(X))]}$$
 
\begin{itemize}
\item In the numerator:
\end{itemize}
 
$$\begin{aligned}
& \mathbb{E}[\tilde{D} \tilde{Y}^{\text{obs}}]=\mathbb{E}[(D-e(X))(Y^{\text{obs}}-\mathbb{E}[Y^{\text{obs}} \mid X])] \\
& =\mathbb{E}[\mathbb{E}[(D-e(X))(Y^{\text{obs}}-\mathbb{E}[Y^{\text{obs}} \mid X]) \mid X]]
\end{aligned}$$
 
by the law of iterated expectations
 
$$=\mathbb{E}[\mathbb{E}[(D-e(X)) Y^{\text{obs}} \mid X]]-\mathbb{E}[\mathbb{E}[(D-e(X)) \mathbb{E}[Y^{\text{obs}} \mid X] \mid X]]$$
 
Here, the second term cancels:
 
$$\begin{aligned}
\mathbb{E}[\mathbb{E}[(D-e(X)) \mathbb{E}[Y^{\text{obs}} \mid X] \mid X]] & =\mathbb{E}[\mathbb{E}[Y^{\text{obs}} \mid X] \underbrace{\mathbb{E}[D-e(X) \mid X]}_{=0}] \\
& =0
\end{aligned}$$
 
Hence: $\mathbb{E}[\tilde{D} \tilde{Y}^{\text{obs}}]=\mathbb{E}[\mathbb{E}[(D-e(X)) Y^{\text{obs}} \mid X]]$
 
$$\begin{aligned}
& =\mathbb{E}[\mathbb{E}[(D-e(X))(Y(1) D+(1-D) Y(0) \mid X]] \\
& =\mathbb{E}[\mathbb{E}[(D-e(X)) D \tau(X) \mid X]-\mathbb{E}[(D-e(X)) Y(0) \mid X]]
\end{aligned}$$
 
using ignorability:
 
$$\begin{aligned}
& =\mathbb{E}[\tau(X) \mathbb{E}[(D-e(X)) D \mid X]-Y(0) \underbrace{{\mathbb{E}[D-e(X) \mid X]}}_{=0}] \\
& =\mathbb{E}[\tau(X) \mathbb{E}[D^2-e(X) D \mid X]] \\
& =\mathbb{E}[\tau(X)(\mathbb{E}[D^2 \mid X]-e(X) \mathbb{E}[D \mid X])] \\
& =\mathbb{E}[\tau(X)(\mathbb{E}[D \mid X]-e(X)^2)] \\
& =\mathbb{E}[\tau(X)(e(X)-e(X)^2)] \\
& =\mathbb{E}[\tau(X) e(X)(1-e(X))]
\end{aligned}$$
 
\begin{itemize}
\item In the denominator:
\end{itemize}
 
$$\begin{aligned}
\mathbb{E}[\tilde{D}^2] & =\mathbb{E}[(D-e(X))^2] \\
& =\mathbb{E}[D^2-2 D e(X)-e(X)^2] \\
& =\mathbb{E}[\mathbb{E}[D-2 D e(X)+e(X)^2 \mid X]] \\
& =\mathbb{E}[e(X)-2 e(X) \mathbb{E}[D \mid X]+e(X)^2] \\
& =\mathbb{E}[e(X)-2 e(X)^2+e(X)^2] \\
& =\mathbb{E}[e(X)-e(X)^2]=\mathbb{E}[e(X)(1-e(X))]
\end{aligned}$$
 
Hence:
 
$$\delta_{\text{res}}=\frac{\mathbb{E}[\tau(X) e(X)(1-e(X))]}{\mathbb{E}[e(X)(1-e(X))]}$$

\subsection{Weights associated to the unexplained part $\delta_\text{res}$}
\label{proof:residuals_on_residuals_weights}

Under ignorability, within the family $\mathcal{F}$, the unexplained portion is expressed as follows: choosing a weighting of the reference outcomes of the form
$Y(1)w(X) + Y(0)(1-w(X))$ leads to the unexplained part:

$$\mathbb{E}\left[\tau(X)\left(\frac{(1-\pi)(1-w(X))\,e(X) + \pi w(X)(1-e(X))}{(1-\pi)\pi}\right)\right]$$

Let us re-write the numerator to isolate $w(X)$:

$$(1-\pi)(1-w(X))\,e(X) + \pi w(X)(1-e(X))$$
$$= (1-\pi)e(X) - w(X)(1-\pi)e(X) + \pi w(X)(1-e(X))$$
$$= (1-\pi)e(X) + w(X)\Big(\pi(1-e(X)) - (1-\pi)e(X)\Big)$$
$$= (1-\pi)e(X) + w(X)\Big(\pi - \pi e(X) - e(X) + \pi e(X)\Big)$$
$$= (1-\pi)e(X) + w(X)\big(\pi - e(X)\big)$$

Since:
$$\delta_{res} = \frac{\mathbb{E}[\tau(X) e(X)(1-e(X))]}{\mathbb{E}[e(X)(1-e(X))]}$$

We can derive the weights $w(X)$ associated with the unexplained part $\delta_\text{res}$. Assuming that $\delta_{res}$ is associated with a reference outcome defined by $Y(1)w^\star(X) + Y(0)(1-w^\star(X))$, then the weights $w^\star(X)$ satisfy:

$$\frac{(1-\pi)e(X) + w^\star(X)(\pi - e(X))}{(1-\pi)\pi} = \frac{e(X)(1-e(X))}{\mathbb{E}[e(X)(1-e(X))]}$$

$\iff$
$$(1-\pi)e(X) + w^\star(X)(\pi - e(X)) = \frac{\pi(1-\pi)\,e(X)(1-e(X))}{\mathbb{E}[e(X)(1-e(X))]}$$

$\iff$
$$w^\star(X)(\pi - e(X)) = \frac{\pi(1-\pi)\,e(X)(1-e(X))}{\mathbb{E}[e(X)(1-e(X))]} - (1-\pi)e(X)$$

$\iff$

$$w^\star(X) = \frac{1}{\pi - e(X)}\left[\frac{\pi(1-\pi)\,e(X)(1-e(X))}{\mathbb{E}[e(X)(1-e(X))]} - (1-\pi)e(X)\right]$$

provided that $\pi - e(X) \neq 0$

$$= \frac{(1-\pi)e(X)}{\big(\pi - e(X)\big)\,\mathbb{E}[e(X)(1-e(X))]} \times \Big[\pi(1-e(X)) - \mathbb{E}[e(X)(1-e(X))]\Big]$$

$$w^\star(X)= \frac{(1-\pi)e(X)}{\big(\pi - e(X)\big)\,C} \times \Big[\pi(1-e(X)) - C\Big], \quad C = \mathbb{E}[e(X)(1-e(X))]$$

\newpage
\bibliographystyle{plainnat}
\bibliography{bibnew}
\end{document}